\documentclass{IEEEtran}

\usepackage[english]{babel}
\usepackage{amsmath}
\usepackage{amssymb}
\usepackage{amsthm}
\usepackage{amstext}
\usepackage{graphicx}
\usepackage{float}
\usepackage{booktabs}
\usepackage{microtype}
\usepackage{algorithm}
\usepackage[noend]{algorithmic}
\usepackage{url}
\usepackage{cite}
\usepackage{color}
\usepackage{xcolor}
\usepackage{bm}
\UseRawInputEncoding
\addtocounter{MaxMatrixCols}{10}

\newtheorem{theorem}{Theorem}

\newtheorem{lemma}[theorem]{Lemma}
\newtheorem{corollary}[theorem]{Corollary}

\newtheorem{example}{Example}

\long\def\symbolfootnote[#1]#2{\begingroup
\def\thefootnote{\fnsymbol{footnote}}\footnote[#1]{#2}\endgroup}
\IEEEoverridecommandlockouts

\title{A Generalization of Array Codes with Local Properties and Efficient
Encoding/Decoding}

\author{Hanxu Hou,~\IEEEmembership{Member,~IEEE}, Yunghsiang S. Han,~\IEEEmembership{Fellow,~IEEE}, Patrick P. C. Lee,~\IEEEmembership{Senior Member,~IEEE}, You Wu, Guojun Han,~\IEEEmembership{Senior Member,~IEEE},
and Mario Blaum,~\IEEEmembership{Life Fellow,~IEEE}
\thanks{This paper was presented in part at the IEEE Global Communications
Conference (GLOBECOM), 2020 \cite{wu2020}.
H. Hou is with the School of Electrical Engineering \& Intelligentization, Dongguan University of Technology~(E-mail: houhanxu@163.com).
Y. S. Han is with the Shenzhen Institute for Advanced Study, University of Electronic Science and Technology of China~(E-mail: yunghsiang@gmail.com).
P. P. C. Lee is with Department of Computer Science and Engineering, The Chinese University of Hong Kong~(E-mail: pclee@cse.edu.cuhk.hk).
Y. Wu is with Beijing Didi Infinity Technology and Development Co., Ltd.~(E-mail: 278008313@qq.com)
G. Han is with the School of Information Engineering, Guangdong University of Technology~(E-mail: gjhan@gdut.edu.cn).
M. Blaum is with IBM Research Division-Almaden~(E-mail: mblaum@hotmail.com).
This work was partially supported by the National Key R\&D Program of
China (No. 2020YFA0712300), the National Natural Science Foundation of China (No. 62071121, 61871136),
Basic Research Enhancement Program of China under Grant 2021-JCJQ-JJ-0483 and
Research Grants Council of HKSAR (AoE/P-404/18).}}

\begin{document}

%\bibliographystyle{plain}

%\markboth{IEEE Transactions on Information Theory}%
%{Submitted paper}

\maketitle

\begin{abstract}
An {\em $(n,k)$ recoverable property} array code is composed of $m\times
n$ arrays such that any $k$ out of $n$ columns suffice to retrieve all
the information symbols, where $n>k$. Note that maximum distance separable (MDS) array code is a
special $(n,k)$ recoverable property array code of size $m\times n$ with the number of
information symbols being $km$. Expanded-Blaum-Roth (EBR) codes and
Expanded-Independent-Parity (EIP) codes are two classes of $(n,k)$ recoverable property array codes
that can repair any one symbol in a column by locally accessing some other
symbols within the column, where the number of symbols $m$ in a column is a
prime number.  By generalizing the constructions of EBR
and EIP codes, we propose new $(n,k)$ recoverable property array
codes, such that any one symbol can be locally recovered and the number of
symbols in a column can be not only a prime number but also a power
of an odd prime number.  Also, we present an
efficient encoding/decoding method for the proposed generalized EBR (GEBR)
and generalized EIP (GEIP) codes based on the LU factorization of a
Vandermonde matrix. We show that the proposed decoding method has less
computational complexity than existing methods.  Furthermore, we show that the
proposed GEBR codes have both a larger minimum symbol distance and a larger
recovery ability of erased lines for some parameters when compared to EBR
codes.  We also present a necessary and sufficient condition of enabling EBR codes
to recover any $r$ erased lines of a slope for
any parameter $r$, which was an open problem in \cite{BR2019}.
Moreover, we show that EBR codes can recover any $r$ consecutive erased lines
of any slope for any parameter $r$.
%Moreover, we can obtain MDS array codes that have both efficient repair bandwidth for any failed column and local repair property by choosing a well-designed check matrix and suitable parameters.
\end{abstract}

\begin{IEEEkeywords}
Array codes, Expanded-Blaum-Roth codes, Expanded-Independent-Parity codes, local repair, efficient encoding/decoding.
\end{IEEEkeywords}

\section{Introduction}

Modern distributed storage systems require data redundancy to maintain data
availability and durability in the presence of failures.  Two major redundancy
mechanisms are replication and erasure coding.  Compared to replication,
erasure coding can deliver higher data reliability with much lower storage
overhead.
%With erasure coding, a data file is divided into $k$ information symbols of
%the same size, which are encoded to generate $r$ parity symbols.  Both $k$
%information symbols and $r$ parity symbols are stored in different nodes (or
%disks) of a storage system to protect data against node failures.

There are many constructions of erasure correcting codes.  In this work, we focus on
{\em array codes}, which are a class of  erasure correcting codes with only XOR and
cyclic-shift operations being involved in the coding process.  Array codes
have been widely used in storage systems, such as the Redundant Array of
Independent Disk (RAID) \cite{Patterson1989}.  Consider an array code of
size $m\times n$ elements, in which each element stores one {\em symbol} in
the array code. Among the $n$ columns, the first $k$ columns store $m\times
k$ {\em information} symbols to form $k$ information columns, and the
remaining $r=n-k$ columns store $m\times r$ {\em parity} symbols, encoded from the
$m\times k$ information symbols, to form $r$ parity columns. The value of $m$
depends on the code construction, and the $m$ symbols in each column are
stored in the same disk (or node) of a storage system.

{\em Maximum distance separable (MDS)} array codes are a special class of
array codes, where any $k$ out of the $n$ columns can retrieve all
$m\times k$ information symbols stored in the $k$ information columns (i.e.,
providing fault tolerance against any $r$ disk failures). More generally, we
define {\em $(n,k)$ recoverable property array codes} as the $m\times n$ array codes such that
we can recover all the information symbols from any $k$ out of the $n$ columns, where
the number of information symbols is no larger than $km$. When the number of information
symbols is $km$, $(n,k)$ recoverable property array codes are reduced to MDS array codes.
There are many
existing MDS array codes in the literature, and most of them are designed to
tolerate two or three failed columns. For example, EVENODD
\cite{Blaum1995,Hou2018} and RDP \cite{Corbett2004Ro} are two important codes
that can correct double disk failures. STAR codes \cite{huang2008star,hou2021star} and
triple-fault-tolerance codes \cite{Hou2017} can correct three disk failures.
Examples of array codes that can tolerate four or more column failures include
Generalized RDP codes \cite{Blaum2006A}, Independent-Parity (also called
generalized EVENODD) codes \cite{blaum2002evenodd}, Blaum-Roth (BR) codes
\cite{Blaum1993}, the codes in \cite{hou2014}, and Rabin-like codes
\cite{feng2005,hou2018a}.
	
To minimize the storage overhead, it is important
to design codes with the larger length for a given overhead. Modern distributed
storage systems often store the data files that
are geographically distributed across nodes, racks, and data centers. Data
should be accessible even if some nodes, racks, or data centers are offline. This motivates designing storage codes that can locally recover
single-symbol failure and quickly recover large correlated failures such as multi-node
failure, rack failure and data center failure, and have fast encoding and decoding
algorithms.
%In addition to a column failure (i.e., all $m$ symbols are failed in the
%failure column), another failure pattern is that one symbol of a column is failed.
Recently, Expanded-Blaum-Roth (EBR) \cite{Blaum2019,BR2019} codes and
Expanded-Independent-Parity (EIP) codes \cite{BR2019} extend BR codes
\cite{Blaum1993} and Independent-Parity codes \cite{blaum2002evenodd},
respectively, and propose to tolerate any $r$ column failures and locally
repair one failed symbol within any column (called \emph{local repair
property}) by adding some parity symbols into each column.  This improves the
performance of repairing a failed symbol, as the repair can be locally done
within a column without accessing the symbols in other columns.
In addition, EBR codes can recover some erased lines of a slope.
Therefore, one possible application of EBR codes and EIP codes is the sectors or pages failures
in a device, like locally recoverable codes (LRC) \cite{2014A}.
Another possible application is in large-scale distributed storage. We need to explicitly
deal with significant correlated failures, such as rack failures, data center failures, and some other
correlated failures. EBR codes can quickly recover some erased lines of a slope that can naturally
be employed in large-scale distributed storage to recover correlated failures,
i.e., some correlated nodes are erased in a way corresponding to the erased lines of a slope.

\subsection{Basics of EBR and EIP Codes}

An EBR code is represented by an $m\times m$ array, where $m=k+r$ and $m>2$ is a
prime number.  It stores $\alpha\times k$ information symbols in the $k$
information columns with $\alpha$ information symbols each, for some $\alpha <
m$, and uses the $\alpha\times k$ sub-array of information symbols as input
for encoding.  Specifically, for $i=0,1,\ldots,m-1$ and $j=0,1,\ldots,m-1$,
let $a_{i,j}\in \mathbb{F}_q$ be the element in row $i$ and column $j$ of the
$m\times m$ array, where $q$ is a power of 2.
%and the $k\alpha$ information symbols are the symbols in the entries with
%$i=0,1,\ldots,\alpha-1$ and $j=0,1,\ldots,k-1$.
%Given the $m\times k$ information array, the parity symbol $a_{p-1,j}$ in the last row of  information column $j$, $0\le j\le k-1$,  are computed by summing all $p-1$ information bits in the same column, i.e.,
%\begin{equation}
%a_{(p-1),j} = \sum_{u=0}^{p-2} a_{u,j}.
%\label{eq:sum}
%\end{equation}
%Furthermore,
For $j=0,1,\ldots,k-1$, the $m$ symbols $a_{0,j},a_{1,j},\ldots,a_{m-1,j}$ in
column $j$ are represented as an {\em information polynomial}
\[
a_j(x)=a_{0,j}+a_{1,j}x+\ldots +a_{m-1,j}x^{m-1}
\]
over the quotient ring $\mathbb{F}_q[x]/(1+x^m)$.  Given the $\alpha$
information symbols $a_{0,j},a_{1,j},\ldots,a_{\alpha-1,j}$, we compute
$m-\alpha$ symbols $a_{\alpha,j},a_{\alpha+1,j},\ldots,a_{m-1,j}$ for local
repair, such that the polynomial $a_j(x)$ is a multiple of $(1+x)g(x)$, where
$g(x)$ is a factor of $1+x+\cdots+x^{m-1}$.  Similarly, the $m$ symbols in
column $j$ with $j=k,k+1,\ldots,m-1$ are represented as a {\em parity
polynomial} $a_j(x)$ over $\mathbb{F}_q[x]/(1+x^m)$.  The relationship between
the information polynomials and the parity polynomials is given as
\[
\mathbf{H}_{r\times m}\cdot\begin{bmatrix}
a_0(x) & a_1(x) & \cdots & a_{m-1}(x)
\end{bmatrix}^T=\mathbf{0}^T,
\]
where $\mathbf{H}_{r\times m}$ is the $r\times m$ parity-check matrix
\begin{equation}
\mathbf{H}_{r\times m}=
\begin{bmatrix}
 1&1&1 & \cdots & 1\\
 1 & x & x^{2}& \cdots & x^{m-1}\\
 \vdots& \vdots& \vdots & \ddots  & \vdots\\
 1& x^{r-1}& x^{2(r-1)}& \cdots & x^{(r-1)(m-1)}\\
 \end{bmatrix},
\label{eq:matrixH}
\end{equation}
and $\mathbf{0}^T$ is an all-zero column of length $r$.  In solving the above
linear equations, all the $r$ parity polynomials are multiples of
$(1+x)g(x)$.  The resulting codes with the parity-check matrix in
Eq.~\eqref{eq:matrixH} are denoted by EBR$(m,r,q,g(x))$.

An EIP code is an $m\times (m+r)$ array, where $m=k$ and $m$ is a prime
number.  It stores $\alpha\times m$ information symbols in $m$ columns with
$\alpha$ information symbols each, for some $\alpha < m$, and uses the
$\alpha\times m$ sub-array of information symbols for encoding.
Let $a_{i,j}$ be the element in row $i$ and column $j$, where
$i=0,1,\ldots,m-1$ and $j=0,1,\ldots,m+r-1$.  For $j=0,1,\ldots,m-1$, given
the $\alpha$ information symbols $a_{0,j},a_{1,j},\ldots,a_{\alpha-1,j}$, we
compute $m-\alpha$ symbols $a_{\alpha,j},a_{\alpha+1,j},\ldots,a_{m-1,j}$
for local repair, such that the information polynomial
\[
a_j(x)=a_{0,j}+a_{1,j}x+\ldots +a_{m-1,j}x^{m-1}
\]
is in $\mathbb{F}_q[x]/(1+x^m)$ and is a multiple of $(1+x)g(x)$, where $g(x)$
is a factor of $1+x+\cdots+x^{m-1}$.  For $j=m,m+1,\ldots,m+r-1$, the
parity polynomials $a_j(x)$ representing the $m$ symbols stored in
column $j$ are computed by
\begin{align*}
\small
&\begin{bmatrix}
a_{m}(x) &a_{m+1}(x) &\cdots & a_{m+r-1}(x)\\
\end{bmatrix}\\
=&\begin{bmatrix}
a_{0}(x) &a_{1}(x) &\cdots & a_{m-1}(x)\\
\end{bmatrix} \cdot \\
&\begin{bmatrix}
1 &1 &\cdots & 1\\
1 &x &\cdots & x^{r-1}\\
\vdots &\vdots &\ddots & \vdots\\
1 &x^{m-1} &\cdots & x^{(r-1)(m-1)}\\
\end{bmatrix}.
\end{align*}

The above code is denoted by EIP$(m,r,q,g(x))$.

\subsection{Contributions}

In this paper, we propose a generalization that can be used to
construct $(n,k)$ recoverable property array codes with new parameters.  The
following are our main contributions:
\begin{enumerate}
\item
First, we give constructions of generalized EBR (GEBR) codes and generalized
EIP (GEIP) codes that can support
more parameters when compared to EBR codes and EIP codes, respectively.
We show that the $m\times (m=n=k+r)$ GEBR codes %with $g(x)=1$
satisfy the $(n,k)$ recoverable property (i.e., all
the information symbols can be reconstructed from any $k$ out of $n=m$
columns) if $m$ is a power of an odd prime.  The $m\times m$ EBR
codes \cite{BR2019} %with $g(x)=1$
are a special case of our GEBR codes %with $g(x)=1$
with $m$ a prime number.
\item
Second, we present an efficient decoding method for GEBR and GEIP codes
based on the LU factorization of a Vandermonde matrix. We show that the
proposed LU decoding method has less complexity than existing
methods.
\item
Third, we show that GEBR codes have a larger minimum symbol distance than EBR
codes for some parameters. We also show that GEBR codes can recover more
erased lines than EBR codes for some parameters. In addition, we present a necessary
and sufficient condition of recovering any $r$ erased lines of slope
$i$ for $0\leq i\leq r-1$ for any $r$ such that $1\leq r\leq m-3$ for EBR codes.
Note that the lines of slope $i$ are taken toroidally, and an erased line of slope
$i$ means that the $m$ symbols in a line of slope $i$ are erased.
\end{enumerate}

LRC codes~\cite{2012On,huang2013pyramid,2014A} can
also locally repair a single-symbol. An example is the code used by
Facebook in its f4 storage system~\cite{sathiamoorthy2013}.
%LRC \cite{huang2013pyramid,2014A} can also locally repair any one symbol and is used in Facebook in
%its f4 storage system \cite{sathiamoorthy2013}.
More general construction of LRC is called grid-like codes (with global parity symbols) or
product codes (without global parity symbols) \cite{gopalan2017maximally}.
Some constructions of LRC are given in \cite{rawat2016locality,barg2017locally,kong2021new}.
The main difference between LRC and
ours is as follows. LRC contains both local parity symbols and global
parity symbols, while all the parity symbols in our codes are local
parity symbols. Each symbol in our codes can be repaired by either
some symbols in the same column or the symbols along a line, while
not in LRC \cite{2014A}. Please refer to Section \ref{sec:com} for
the detailed comparison of LRC, product codes, and the proposed codes.

\subsection{Paper Organization}
The rest of the paper is organized as follows.
Section~\ref{sec:framework} gives the generalized coding method.
Section~\ref{sec:constr} presents the construction of GEBR codes based on the
coding method and proposes the LU decoding method.
Section~\ref{sec:geip} presents the construction of GEIP codes.
Section~\ref{sec:distance} discusses the minimum symbol distance for the
proposed codes.
Section~\ref{sec:lines} shows that GEBR codes can recover some erased lines.
Section~\ref{sec:com} compares GEBR codes with other related codes.
%Section \ref{sec:optimal} presents constructions of MDS array codes that have efficient repair bandwidth for any column and have local repair property.
Section~\ref{sec:con} concludes the paper.

\section{Generalized Coding Method of Array Codes with Local Properties}
\label{sec:framework}

In this section, we present a coding method for array codes that can encode
an $\alpha\times k$ sub-array of information symbols into an $m\times (k+r)$
array, where each element in the array is in the finite field
$\mathbb{F}_{q}$, $q$ is a power of 2, $m=p\tau$, $p$ is a prime number, and
$\alpha,k,r,\tau$ are positive integers with $\alpha\leq (p-1)\tau$.  The
primary objective of the coding method is to extend the constructions of
EBR and EIP codes to support much more parameters.  In particular, EBR codes
can be viewed as a special construction of the proposed coding method with
$\tau=1$ and $k+r=p$, while EIP codes are an explicit construction of the
coding method with $\tau=1$ and $k=p$.

Let $s_{i,j}\in \mathbb{F}_{q}$ be the element of the $m\times (k+r)$ array in
row $i$ and column $j$, where $i=0,1,\ldots,m-1$ and $j=0,1,\ldots,k+r-1$. The
$\alpha k$ information symbols are $s_{i,j}$ with $i=0,1,\ldots,\alpha-1$ and
$j=0,1,\ldots,k-1$, and the other $m(k+r)-\alpha k$ elements of the array are
parity symbols.

For $j=0,1,\ldots,k+r-1$, we represent the $m$ symbols stored in column $j$
(i.e., $s_{0,j},s_{1,j},\ldots,s_{m-1,j}$) by a polynomial $s_j(x)$ of degree
$m -1$ over the ring $\mathbb{F}_{q}[x]$, i.e.,
\[
s_j(x)=s_{0,j}+s_{1,j}x+s_{2,j}x^2+\cdots+s_{m-1,j}x^{m-1},
\]
where $s_j(x)$ with $j=0,1,\ldots,k-1$ is an information polynomial and
$s_j(x)$ with $j=k,k+1,\ldots,k+r-1$ is a parity polynomial.  Let
$\mathcal{R}_{m}(q)=\mathbb{F}_{q}[x]/(1+x^{m})$ be the ring of polynomials
modulo $1+x^{m}$ with coefficients in $\mathbb{F}_{q}$. We observe that
multiplication by $x^i$ in $\mathcal{R}_{m}(q)$ can be interpreted as a
\emph{cyclic shift}, and hence it does not involve finite field arithmetic
nor XOR operations.

Given $\alpha$ information symbols $s_{0,j},s_{1,j},\cdots,s_{\alpha-1,j}$,
we need to determine $m-\alpha$ parity symbols
$s_{\alpha,j},s_{\alpha+1,j},\cdots,s_{m-1,j}$, where $j=0,1,\ldots,k-1$.
Let $g(x)$ be a polynomial with coefficients in $\mathbb{F}_{q}$ such that
$g(x)$ divides $1+x^{\tau}+\cdots+x^{(p-1)\tau}$ and
$\gcd (g(x), 1+x^{\tau})=1$.  Let  $1+x^\tau+\cdots+x^{(p-1)\tau}=g(x)h(x)$.
%We also require that $\gcd (g(x), h(x))=1$ and $\gcd (1+x^\tau, h(x))=1$.
Note that $g(x)$ may not be an irreducible polynomial so that we can factorize
$g(x)$ as a product of powers of irreducible polynomials over
$\mathbb{F}_{q}$, i.e.,
\[
g(x)=(f_1(x))^{\ell_1}\cdot (f_2(x))^{\ell_2} \cdots (f_t(x))^{\ell_t},
\]
where $\ell_i\geq 1$ for $i=1,2,\ldots,t$ and $\deg (f_i(x))\geq \deg (f_j(x))$ for $i>j$.

Let $\mathcal{C}_{p\tau}(g(x),\tau,q,d)$ be the cyclic code of length $m=p\tau$ over
$\mathbb{F}_{q}$ with generator polynomial $(1+x^\tau)g(x)$ and minimum
distance $d$.  For $j=0,1,\ldots,k-1$, we create $m-\alpha$ parity symbols
$s_{\alpha,j},s_{\alpha+1,j},\cdots,s_{m-1,j}$ by encoding $\alpha$
information symbols $s_{0,j},s_{1,j},\cdots,s_{\alpha-1,j}$, such that the
polynomial $s_j(x)$ is in $\mathcal{C}_{p\tau}(g(x),\tau,q,d)$.

Given $k$ information polynomials $s_0(x),s_1(x),\ldots,s_{k-1}(x)$,
we can compute the $r$ parity polynomials $s_k(x),s_{k+1}(x)$, $\ldots,s_{k+r-1}(x)$ by taking
the product
\begin{eqnarray}
&&\begin{bmatrix}
s_k(x) & s_{k+1}(x) & \cdots & s_{k+r-1}(x)\\
\end{bmatrix}\nonumber \\
&=&\begin{bmatrix}
s_0(x) & s_{1}(x) & \cdots & s_{k-1}(x)\\
\end{bmatrix}\cdot \mathbf{P}_{k\times r}
\label{eq:enc-matrix}
\end{eqnarray}
with operations performed in $\mathcal{R}_{p\tau}(q)$, where
$\mathbf{P}_{k\times r}$ is the $k\times r$ encoding matrix which
is the remainder sub-matrix of the systematic generator matrix by deleting
the identity matrix.
We can also compute the $r$ parity polynomials by
\begin{equation}
\begin{bmatrix}
s_0(x) & s_{1}(x) & \cdots & s_{k+r-1}(x)\\
\end{bmatrix}\cdot \mathbf{H}^T_{r\times (k+r)}=\mathbf{0}
\label{eq:check-matrix}
\end{equation}
over $\mathcal{R}_{p\tau}(q)$, where $\mathbf{H}_{r\times (k+r)}$
is an $r\times (k+r)$ \emph{parity-check matrix}.

Note that $\mathcal{C}_{p\tau}(g(x),\tau,q,d)$ is an {\em ideal}
in $\mathcal{R}_{p\tau}(q)$, because
$\forall c(x)\in \mathcal{R}_{p\tau}(q), \forall s(x)\in \mathcal{C}_{p\tau}(g(x),\tau,q,d)$,
we have $c(x)s(x)\in \mathcal{C}_{p\tau}(g(x),\tau,q,d)$.
Recall that $g(x)h(x)=1+x^\tau+\cdots+x^{(p-1)\tau}$.
The polynomial $h(x)$ is called \emph{parity-check polynomial} of
$\mathcal{C}_{p\tau}(g(x),\tau,q,d)$, since the multiplication of any polynomial in
$\mathcal{C}_{p\tau}(g(x),\tau,q,d)$ and $h(x)$ is zero. We show in the next theorem
that $\mathcal{R}_{p\tau}(q)$ is isomorphic to
$\mathbb{F}_{q}[x]/g(x)(1+x^\tau)\times \mathbb{F}_{q}[x]/(h(x))$
under some specific conditions.
\begin{theorem}
When $\gcd (g(x), h(x))=1$ and $\gcd (1+x^\tau, h(x))=1$,
the ring $\mathcal{R}_{p\tau}(q)$ is isomorphic to
$\mathbb{F}_{q}[x]/g(x)(1+x^\tau)\times \mathbb{F}_{q}[x]/(h(x))$.
\label{lm:isom}
\end{theorem}
\begin{proof}
When $\gcd (g(x), h(x))=1$
and $\gcd (1+x^\tau, h(x))=1$, we have $\gcd (g(x)(1+x^\tau), h(x))=1$.
By the Chinese Remainder Theorem, we can find an isomorphism between
$\mathcal{R}_{p\tau}(q)$ and $\mathbb{F}_{q}[x]/(g(x)(1+x^\tau))\times
\mathbb{F}_{q}[x]/(h(x))$.  The mapping $\theta$ is defined by
\[
\theta(a(x))=(a(x)\bmod (g(x)(1+x^\tau)), a(x) \bmod (h(x))),
\]
where $a(x)\in \mathcal{R}_{p\tau}(q)$.

Let $h(x)\bmod g(x)(1+x^{\tau})$ be the remainder of dividing $h(x)$ by $g(x)(1+x^{\tau})$ which is
in $\mathbb{F}_{q}[x]/(g(x)(1+x^\tau))$ and $g(x)(1+x^{\tau})\bmod h(x)$ be the remainder of
dividing $g(x)(1+x^{\tau})$ by $h(x)$ which is in $\mathbb{F}_{q}[x]/(h(x))$.
Since $\gcd (g(x)(1+x^\tau), h(x))=1$, there exists the inverse of $h(x)\bmod g(x)(1+x^{\tau})$
in $\mathbb{F}_{q}[x]/(g(x)(1+x^\tau))$ and denote $(h(x)\bmod g(x)(1+x^{\tau}))^{-1}$ as the
inverse. Similarly, denote $(g(x)(1+x^{\tau})\bmod h(x))^{-1}$ as the inverse of
$g(x)(1+x^{\tau})\bmod h(x)$ in $\mathbb{F}_{q}[x]/(h(x))$.
The inverse mapping $\theta^{-1}$ is
\begin{align*}
\small
&\theta^{-1}(a_1(x),a_2(x))=\big( a_1(x) h(x) (h(x)\bmod (g(x)(1+x^\tau)))^{-1}+\\
&a_2(x) g(x)(1+x^\tau) ((g(x)(1+x^\tau))\bmod (h(x)))^{-1} \big) \bmod (1+x^{p\tau}),
\end{align*}
where $a_1(x)\in \mathbb{F}_{q}[x]/(g(x)(1+x^\tau))$ and $a_2(x)\in \mathbb{F}_{q}[x]/(h(x))$.

Then we have
\begin{small}
\begin{align*}
&\theta(\theta^{-1}(a_1(x),a_2(x)))\\
=&\theta\Big(\big( a_1(x) h(x) (h(x)\bmod (g(x)(1+x^\tau)))^{-1}+a_2(x) g(x)\\
&(1+x^\tau)((g(x)(1+x^\tau))\bmod (h(x)))^{-1} \big) \bmod (1+x^{p\tau})\Big)\\
=&\Big(\Big(\big( a_1(x) h(x)  (h(x)\bmod (g(x)(1+x^\tau)))^{-1}\bmod (g(x)(1+x^{\tau}))\\
&+a_2(x) g(x)(1+x^\tau) ((g(x)(1+x^\tau))\bmod (h(x)))^{-1} \big) \\
&\bmod (1+x^{p\tau})\bmod (g(x)(1+x^{\tau}))\Big),\\
&\Big(\big( a_1(x) h(x) (h(x)\bmod (g(x)(1+x^\tau)))^{-1}\bmod (h(x))+\\
&a_2(x) g(x)(1+x^\tau) ((g(x)(1+x^\tau))\bmod (h(x)))^{-1} \big) \\
&\bmod (1+x^{p\tau})\bmod (h(x))\Big)\Big)\\
=&(a_1(x),a_2(x)),
\end{align*}
\end{small}
and the theorem is proved.
\end{proof}
Note that the result in Lemma 2 in \cite{HOU2019} can be viewed as a special case
of our Theorem \ref{lm:isom} with $g(x)=1$ and $q=2$.
By Theorem \ref{lm:isom}, we can directly obtain that $\mathcal{C}_{p\tau}(g(x),\tau,q,d)$ is
isomorphic to $\mathbb{F}_{q}[x]/(h(x))$, and we give the isomorphism in the next lemma.

\begin{lemma}
When $\gcd (g(x), h(x))=1$ and $\gcd (1+x^\tau, h(x))=1$,
the ring $\mathcal{C}_{p\tau}(g(x),\tau,q,d)$ is isomorphic to $\mathbb{F}_{q}[x]/(h(x))$,
the isomorphism $\theta: \mathcal{C}_{p\tau}(g(x),\tau,q,d) \rightarrow \mathbb{F}_{q}[x]/(h(x))$ is
$\theta (a(x))=a(x)\bmod h(x)$ and the inverse isomorphism $\theta^{-1}: \mathbb{F}_{q}[x]/(h(x))
 \rightarrow \mathcal{C}_{p\tau}(g(x),\tau,q,d)$ is
$\theta^{-1}(a(x))=a(x)\cdot g(x)(1+x^\tau) \cdot (g(x)(1+x^\tau))\bmod (h(x))^{-1} \bmod (1+x^{p\tau})$.
\label{lm:isom2}
\end{lemma}

In the next lemma, we show a necessary condition of a polynomial in the ring $\mathcal{C}_{p\tau}(g(x),\tau,q,d)$.
\begin{lemma}
If the polynomial $s_j(x)=\sum_{i=0}^{m-1}s_{i,j}x^{i}$ is in $\mathcal{C}_{p\tau}(g(x),\tau,q,d)$,
then the coefficients of polynomial $s_j(x)$ satisfy the following equation
\begin{equation}
\sum_{\ell=0}^{p-1}s_{\ell \tau+\mu,j}=0,
\label{eq:coeff}
\end{equation}
where $\mu=0,1,\ldots,\tau-1$.
\label{lm:coeff}
\end{lemma}
\begin{proof}
The proof is similar to that in Theorem 1 in~ \cite{HOU2019}.
%Suppose that $s_j(x)=\sum_{i=0}^{m-1}s_{i,j}x^{i}$ is in $\mathcal{C}_{p\tau}(g(x),q,d)$.
%Then $s_j(x)$ can be written as
%\begin{align*}
%s_j(x)=&a(x)(1+x^{\tau})\bmod (1+x^{p\tau})\\
%=&(a_0+a_{(p-1)\tau})+(a_1+a_{(p-1)\tau+1})x+\cdots+(a_\tau+a_{0})x^{\tau}+\cdots+(a_{p\tau-1}+a_{(p-1)\tau-1})x^{p\tau-1},
%\end{align*}
%where $a(x)$ is a multiple of $g(x)$.
%We can obtain
%\begin{align*}
%s_{\mu,j}=&a_{\mu}+a_{(p-1)\tau+\mu},\\
%s_{\tau+\mu,j}=&a_{\tau+\mu}+a_{\mu},\\
%&\cdots,\\
%s_{(p-1)\tau+\mu,j}=&a_{(p-1)\tau+\mu}+a_{(p-2)\tau+\mu},
%\end{align*}
%for $\mu=0,1,\ldots,\tau-1$. Therefore, we can check that Eq. \eqref{eq:coeff} holds and the lemma is proved.
\end{proof}

When $g(x)=1$, the next lemma shows that the necessary condition given in
Lemma \ref{lm:coeff} is also the sufficient condition.
\begin{lemma}\cite[Theorem 1]{HOU2019}
When $g(x)=1$, the polynomial $s_j(x)=\sum_{i=0}^{m-1}s_{i,j}x^{i}$ is in
$\mathcal{C}_{p\tau}(1,\tau,q,d)$ if and only if Eq. \eqref{eq:coeff} holds.
\label{lm:coeff1}
\end{lemma}
When $g(x)=1$, we have that the weight (the number of non-zero coefficients) of
$s_j(x)\in \mathcal{C}_{p\tau}(1,\tau,q,d)$ is a positive even integer by Lemma \ref{lm:coeff1}.

When $g(x)=1$ and $q=2$, the ring $\mathcal{C}_{p\tau}(1,\tau,2,d)$ has been used
in the literature to give efficient repair for a family of binary MDS array
codes \cite{Hou2017,HOU2019} and to provide new constructions of regenerating
codes with lower computational complexity \cite{hou2019new}.
When $g(x)=1$, $q=2$, and $\tau=1$, the ring is discussed in
\cite{Blaum1993,blaum1996mds,hou2014,hou2016,hou2018a,Hou2018form}.
When $\tau=1$, the ring is used to construct array codes with local properties \cite{BR2019}.

When $g(x)=1$, if we delete the last $\tau$ rows of the $m\times (k+r)$ array
of our coding method, then the obtained $(p-1)\tau\times (k+r)$ array
is reduced to the coding method given in \cite{HOU2019}.

Note that the ring $\mathcal{C}_{p\tau}(1,\tau,2,2)$ is reduced to
a finite field of size $2^{(p-1)\tau}$ if and only if 2 is a primitive element in $\mathbb{Z}_p$
and $\tau=p^i$ for some non-negative integer $i$ \cite{Itoh1991Characterization}.
When $\tau$ is a power of $p$ and $p$ is a prime number such that 2
is a primitive element in $\mathbb{Z}_p$, we have $g(x)=1$, $d=2$ and
$h(x)=1+x^{\tau}+\cdots+x^{(p-1)\tau}$ is an irreducible polynomial in
$\mathbb{F}_{2}[x]$.

\section{Generalized Expanded-Blaum-Roth Codes}
\label{sec:constr}
In this section, we first give the construction of GEBR codes and then propose
the LU decoding method that can be used in the encoding/decoding procedures of GEBR codes.

\subsection{Construction}
The proposed GEBR code is a set of arrays of size $m\times (k+r)$ by encoding
$k\alpha$ information symbols, where $m=p\tau$, $\alpha<m$,
$k+r\leq m$,
$\tau$ is a positive integer, and $p$ is an odd prime number.
The constructed GEBR code is denoted by $\textsf{GEBR}(p,\tau,k,r,q,g(x))$ with
parity-check matrix given as
\begin{equation}
\mathbf{H}_{r\times (k+r)}=
\begin{bmatrix}
 1&1&1 & \cdots & 1\\
 1 & x & x^{2}& \cdots & x^{k+r-1}\\
 \vdots& \vdots& \vdots & \ddots  & \vdots\\
 1& x^{r-1}& x^{2(r-1)}& \cdots & x^{(r-1)(k+r-1)}\\
 \end{bmatrix}.
\label{eq:matrixH2}
\end{equation}
Note that we have more than one solution of
$s_k(x),s_{k+1}(x),\ldots,s_{k+r-1}(x)$
in Eq.~\eqref{eq:check-matrix}. We need to choose one solution such that all $r$
polynomials $s_k(x),s_{k+1}(x),\ldots,s_{k+r-1}(x)$ are in
$\mathcal{C}_{p\tau}(g(x),\tau,q,d)$.
Since we will show the $(n,k)$ recoverable condition of $\textsf{GEBR}(p,\tau,k,r,q,g(x))$ in
Theorem~\ref{lm:mds} and Theorem \ref{thm:mds-gx1}, we assume that the parameters $p,\tau,k,r,q,g(x)$ satisfy
the $(n,k)$ recoverable condition given in Theorem \ref{lm:mds} or Theorem \ref{thm:mds-gx1}, and there exists only one
solution such that all $r$
polynomials $s_k(x),s_{k+1}(x),\ldots,s_{k+r-1}(x)$ are in
$\mathcal{C}_{p\tau}(g(x),\tau,q,d)$.
For general polynomial $g(x)$, we require that $\gcd (g(x), h(x))=1$,
$\gcd (1+x^\tau, h(x))=1$, $1+x^i$ and $h(x)$ are relatively prime
over $\mathbb{F}_{q}[x]$ for $i=1,2,\ldots,k+r-1$, $\textsf{GEBR}(p,\tau,k,r,q,g(x))$
are $(n,k)$ recoverable property codes. Please refer to Theorem~\ref{lm:mds} for the detailed proof.
When $g(x)=1$, let $\tau=\gamma p^{\nu}$, where $\nu\geq 0$, $0< \gamma$ and $\gcd (\gamma,p)=1$.
The codes $\textsf{GEBR}(p,\tau,k,r,q,1)$ are $(n,k)$ recoverable property codes if and only if $k+r\leq p^{\nu +1}$.
Please refer to Theorem \ref{thm:mds-gx1} for the detailed proof.

The encoding procedure is described as follows.
We  first replace each entry $a(x)$ of the parity-check matrix in Eq. \eqref{eq:matrixH2}
by $a(x)\cdot g(x)(1+x^\tau) \cdot (g(x)(1+x^\tau))\bmod (h(x))^{-1} \bmod (1+x^{p\tau})$
that is in $\mathcal{C}_{p\tau}(g(x),\tau,q,d)$ by Lemma \ref{lm:isom2} and then solve the
$r$ polynomials $s_k(x),s_{k+1}(x),\ldots,s_{k+r-1}(x)$
based on the modified parity-check matrix over $\mathcal{C}_{p\tau}(g(x),\tau,q,d)$.

From row $i$ of the parity-check matrix in Eq. \eqref{eq:matrixH2}, where $i=0,1,\ldots,r-1$,
the summation of the $k+r$ symbols in each line of slope $i$ of the $m \times (k+r)$ array is zero,
{i.e.,
\[
s_{\ell,0}+s_{\ell-i,1}+s_{\ell-2i,2}+\cdots+s_{\ell-(k+r-1)i,k+r-1}=0,
\]
for $\ell=0,1,\ldots,m-1$. The indices are taken modulo $m$ throughout the paper unless
otherwise specified. For example, when $i=1$, we have
\[
s_0(x)+xs_1(x)+\cdots+x^{k+r-1}s_{k+r-1}(x)=0 \bmod (1+x^m).
\]
Recall that
\begin{align*}
x^js_j(x)=&x^j(\sum_{\ell=0}^{m-1}s_{\ell,j}x^{\ell})\\
=&s_{m-j,j}+s_{m-j+1,j}x+\cdots+s_{0,j}x^{j}+\\
&s_{1,j}x^{1+j}+\cdots+s_{m-j-1,j}x^{m-1},
\end{align*}
for $j=1,2,\ldots,k+r-1$. We can obtain that the summation of the $k+r$ symbols
$s_{\ell,0},s_{\ell-1,1},\ldots,$ $s_{\ell-(k+r-1),k+r-1}$ in each line of slope $i=1$ is zero.
The following array is an example of $m=5$ and $k+r=4$, where the symbols
$s_{1,0},s_{0,1},s_{4,2},s_{3,3}$ with bold font are in one line of slope $i=1$:
\begin{align*}
\begin{bmatrix}
s_{0,0} & \bm{s_{0,1}} & s_{0,2} & s_{0,3} \\
\bm{s_{1,0}} & s_{1,1} & s_{1,2} & s_{1,3} \\
s_{2,0} & s_{2,1} & s_{2,2} & s_{2,3} \\
s_{3,0} & s_{3,1} & s_{3,2} & \bm{s_{3,3}} \\
s_{4,0} & s_{4,1} & \bm{s_{4,2}} & s_{4,3} \\
\end{bmatrix}.
\end{align*}
}
Note that the code proposed in \cite{Blaum2019}
is a special case as $\textsf{GEBR}(p,\tau=1,k,r=p-k,q,g(x)=1)$, and
$\textsf{GEBR}(p,\tau=1,k,r=p-k,q,g(x))$ is the EBR code in \cite{BR2019}.
In the $p\times p$ array of EBR code in \cite{BR2019}, the summation of
the $p$ symbols along a line of slopes $0,1,\ldots,r-1$ is zero, the polynomial
corresponding to the $p$ symbols in one column is a polynomial in $\mathbb{F}_q[x]/(1+x^p)$
which is a multiple of $g(x)(1+x)$. In the $p\tau\times p\tau$ array of our
$\textsf{GEBR}(p,\tau,k,r=p\tau-k,q,g(x))$, the summation of the symbols in each line of
slope $0,1,\ldots,r-1$ is zero, as like EBR codes. The difference is that the polynomial
corresponding to each column of EBR codes is a multiple of $g(x)(1+x)$, where $g(x)$
is a factor of $1+x+\cdots+x^{p-1}$. While the polynomial corresponding to each column
of $\textsf{GEBR}(p,\tau,k,r=p\tau-k,q,g(x))$ is a multiple of $g(x)(1+x^{\tau})$,
where $g(x)$ is a factor of $1+x^{\tau}+\cdots+x^{(p-1)\tau}$ such that $\gcd (g(x), 1+x^{\tau})=1$.

\subsection{The $(n,k)$ Recoverable Property}
%By \eqref{eq:check-matrix} and \eqref{eq:matrixH2}, we have
%\begin{align}
%\label{eq:computing_parity}
%&\begin{bmatrix}
%1 & 1 & \cdots & 1\\
%1 & x & \cdots & x^{k-1}\\
%\vdots & \vdots & \ddots & \vdots\\
%1 & x^{r-1} & \cdots & x^{(r-1)(k-1)}\\
%\end{bmatrix}\begin{bmatrix}
%s_0(x) \\ s_1(x) \\ \vdots \\ s_{k-1}(x)
%\end{bmatrix}\nonumber\\
%=&\begin{bmatrix}
%1 & 1 & \cdots & 1\\
%x^k & x^{k+1} & \cdots & x^{k+r-1}\\
%\vdots & \vdots & \ddots & \vdots\\
%x^{(r-1)k} & x^{(r-1)(k+1)} & \cdots & x^{(r-1)(k+r-1)}\\
%\end{bmatrix}\begin{bmatrix}
%s_k(x) \\ s_{k+1}(x) \\ \vdots \\ s_{k+r-1}(x)
%\end{bmatrix}.
%\end{align}
%We can compute the $r$ polynomials $s_{k}(x),s_{k+1}(x),\ldots,s_{k+r-1}(x)$ by solving \eqref{eq:computing_parity} over $\mathcal{R}_{m}(q)$ if the determinant of the $r\times r$ matrix on the right-hand side of \eqref{eq:computing_parity} is relatively prime to $h(x)$.

When we say that a code is $(n,k)$ recoverable or satisfies the $(n,k)$ recoverable property, it means that
the code can recover up to any $r$ erased columns out of the $k+r$ columns.
One necessary and sufficient $(n,k)$ recoverable condition of
$\textsf{GEBR}(p,\tau,k,r,q,g(x))$ with $\gcd (g(x), h(x))=1$
and $\gcd (1+x^\tau, h(x))=1$ is given in the next theorem.
%The determinant can be written as the multiplication of $r$ factors $1+x^i$, where $i\in\{1,2,\ldots,k+r-1\}$. Therefore, we have the following lemma.
\begin{theorem}
When $\gcd (g(x), h(x))=1$ and $\gcd (1+x^\tau, h(x))=1$,
we can compute all the $k\alpha$ information symbols
from any $k$ out of $k+r$ polynomials $s_{0}(x),s_{1}(x),\ldots,s_{k+r-1}(x)$,
if and only if, the two polynomials $1+x^i$ and $h(x)$ are relatively prime
over $\mathbb{F}_{q}[x]$, where $i=1,2,\ldots,k+r-1$.
\label{lm:mds}
\end{theorem}
\begin{proof}
By Lemma \ref{lm:isom2}, it is sufficient to show that the determinant of any $r\times r$ sub-matrix
of $\mathbf{H}_{r\times (k+r)}$ in Eq. \eqref{eq:matrixH2} is invertible over
$\mathbb{F}_q[x]/h(x)$. Any $r\times r$
sub-matrix of $\mathbf{H}_{r\times (k+r)}$ is a Vandermonde matrix,
the determinant can be written as the multiplication of power of $x$ and $r$
different factors
$1+x^i$, where $i\in\{1,2,\ldots,k+r-1\}$.
Note that the coefficient of constant term of $h(x)$ is non-zero, we have $\gcd(x^j,h(x))=1$
for any positive integer $j$.
%In other words,
The determinant can be viewed as a polynomial in
$\mathbb{F}_q[x]/h(x)$ after modulo $h(x)$, and is invertible over
the ring $\mathbb{F}_q[x]/h(x)$. Therefore, we can compute all the
$k\alpha$ information symbols
from any $k$ out of $k+r$ polynomials, if and only if, $1+x^i$
is invertible over $\mathbb{F}_q[x]/h(x)$
for all $i=1,2,\ldots,k+r-1$.
\end{proof}

If $\tau$ is a power of 2 and $g(x)=1$, we have
\[
h(x)=1+x^{\tau}+\ldots+x^{(p-1)\tau}=(1+x+\ldots+x^{p-1})^\tau.
\]
We can check that $\gcd (g(x)=1, h(x))=1$ and
$\gcd (1+x^\tau, h(x)=(1+x+\ldots+x^{p-1})^\tau)=1$.
The $(n,k)$ recoverable condition in Theorem \ref{lm:mds} is reduced to
that $1+x^i$ and $1+x+\ldots+x^{p-1}$ are relatively prime over
$\mathbb{F}_{q}[x]$ when $\tau$ is a power of 2.
Note that $1+x^i$ and $1+x+\ldots+x^{p-1}$ are relatively prime over
$\mathbb{F}_{q}[x]$ for $i=1,2,\ldots,p-1$ \cite{Blaum1993}.
Therefore, when $\tau$ is a power of 2 and $g(x)=1$,
$\textsf{GEBR}(p,\tau,k,r,q,g(x))$ are $(n,k)$ recoverable if $k+r\leq p$
and are not $(n,k)$ recoverable if $k+r>p$.

According to Theorem \ref{lm:mds}, $\textsf{GEBR}(p,\tau,k,r,q,g(x))$ are
$(n,k)$ recoverable if and only if $\gcd (h(x),1+x^i)=1$ for $i=1,2,\ldots,k+r-1$.
In the following, we present an equivalent necessary and sufficient $(n,k)$ recoverable condition.
\begin{lemma}
The codes $\textsf{GEBR}(p,\tau,k,r,q,g(x))$ are $(n,k)$ recoverable if and only if the following equation
\begin{equation}
(1+x^i)s(x)=c(x)\bmod (1+x^{p\tau})
\label{eq:mds-cond}
\end{equation}
has a unique solution $s(x)\in\mathcal{C}_{p\tau}(g(x),\tau,q,d)$, given that $c(x)\in\mathcal{C}_{p\tau}(g(x),\tau,q,d)$
and $i\in \{1,2,\ldots,k+r-1\}$.
\label{lm:mds-equi}
\end{lemma}
\begin{proof}
First we prove that Eq. \eqref{eq:mds-cond}
has a unique solution $s(x)\in\mathcal{C}_{p\tau}(g(x),\tau,q,d)$, given that $c(x)\in\mathcal{C}_{p\tau}(g(x),\tau,q,d)$,
if and only if, $1+x^i$ and $h(x)$ are relatively prime
over $\mathbb{F}_{q}[x]$ for  $i\in \{1,2,\ldots,k+r-1\}$.

\noindent $(\Longrightarrow)$   Since both $s(x)$ and $c(x)$ are in $\in\mathcal{C}_{p\tau}(g(x),\tau,q,d)$, we have
\begin{eqnarray}
&&(1+x^i)s(x)=c(x)\bmod (1+x^{p\tau})\nonumber\\
&\leftrightarrow&	(1+x^i)a(x)(1+x^\tau)g(x)=b(x)(1+x^\tau)g(x) \nonumber\\
&&\bmod (1+x^{\tau})g(x)h(x)\nonumber\\
&\leftrightarrow&	(1+x^i)a(x)=b(x)\bmod h(x).\label{eq:Lemma5-proof}
\end{eqnarray}
Assume that $\gcd(1+x^i, h(x))=d(x)$. Then, $d(x)|b(x)$. Since $c(x)$ is any element in $\mathcal{C}_{p\tau}(g(x),\tau,q,d)$, this is possible only when $d(x)=1$.

\noindent $(\Longleftarrow)$ Note that, in Eq.~\eqref{eq:Lemma5-proof},  both $\deg(a(x))$ and $\deg(b(x))$ are  less than $\deg(h(x))$. Given any valid $b(x)$, we need to prove that there exists only one solution $a(x)$ for Eq.~\eqref{eq:Lemma5-proof}. Let $1+x^i\bmod h(x)=e(x)$. Since $\gcd(1+x^i,h(x))=1$, we have $\gcd(e(x),h(x))=1$ such that the inverse of $e(x)$ exists. Note that $\deg(a(x))<\deg(h(x))$ and the only solution $a(x)$ for Eq.~\eqref{eq:Lemma5-proof} is $a(x)=e(x)^{-1}b(x)\bmod h(x)$.
By Theorem~\ref{lm:mds}, this completes the proof.
\end{proof}

Next we prove that to determine the condition given in Eq.~\eqref{eq:mds-cond}, we only need to check for the case $c(x)=0$.
\begin{lemma}
\label{lm:null-space}
Eq.~\eqref{eq:mds-cond}
has a unique solution $s(x)\in\mathcal{C}_{p\tau}(g(x),\tau,q,d)$, given that $c(x)\in\mathcal{C}_{p\tau}(g(x),\tau,q,d)$ if and only if Eq.~\eqref{eq:mds-cond}
has a unique solution $s(x)=0$ when $c(x)=0$.
\end{lemma}
\begin{proof}
The ``if" part is obvious such that we only need to prove the ``only if" part. By Eq.~\eqref{eq:mds-cond}, we define the transformation $g_i:\mathcal{C}_{p\tau}(g(x),\tau,q,d)	\longrightarrow \mathcal{C}_{p\tau}(g(x),\tau,q,d)$ as
$$g_i(s(x))=(1+x^i)s(x) \bmod (1+x^{p\tau}),$$ where $s(x)\in \mathcal{C}_{p\tau}(g(x),\tau,q,d)$.
It is easy to see that this transformation is linear. Since $(1+x^i)s(x)=0$ has unique solution $s(x)=0$, the kernel of $g_i$ is of dimension $0$. According to the rank theorem, the dimension of the range of $g_i$ is the same as the dimension of $\mathcal{C}_{p\tau}(g(x),\tau,q,d)$. Hence, $g_i$ is a one-to-one and onto mapping such that, for any $g_i(s(x))=c(x)$,  Eq.~\eqref{eq:mds-cond}
has a unique solution $s(x)\in\mathcal{C}_{p\tau}(g(x),\tau,q,d)$, given that $c(x)\in\mathcal{C}_{p\tau}(g(x),\tau,q,d)$.
\end{proof}

When $g(x)=1$, we can locally recover any one symbol with the minimum parity symbol
(one local parity symbol for $p-1$ data symbols in each data column),
which is practically interesting in storage systems. For example, the LRC implemented in
Facebook \cite{sathiamoorthy2013} employs one local parity symbol in each group.
In the following, we investigate the necessary and sufficient
$(n,k)$ recoverable condition when $g(x)=1$.
%We have checked that when $\tau=3$ and $p=5,7,11$, the polynomial
%$1+x+\ldots+x^{p-1}$ is a factor of $1+x^{\tau}+\ldots+x^{(p-1)\tau}$
%and $\textsf{GEBR}(p,\tau,k,r,q,g(x))$ are MDS codes if $k+r\leq p$
%and are not MDS codes if $k+r>p$.
When $g(x)\neq 1$, the $(n,k)$ recoverable condition of $\textsf{GEBR}(p,\tau,k,r,q,g(x))$
is a subject of future work.

\begin{theorem}
Let $\tau=\gamma p^{\nu}$, where $\nu\geq 0$, $0< \gamma$ and $\gcd (\gamma,p)=1$.
Then, the codes $\textsf{GEBR}(p,\tau,k,r,q,1)$ are $(n,k)$ recoverable if and only if $k+r\leq p^{\nu +1}$.
%If $k+r\leq p^{\nu+1}$, then the codes $\textsf{GEBR}(p,\tau,k,r,q,1)$ are MDS codes.
%Otherwise, if $k+r> p^{\nu+1}$, then the codes $\textsf{GEBR}(p,\tau,k,r,q,1)$ are not MDS codes.
\label{thm:mds-gx1}
\end{theorem}
\begin{proof}
We need to show that the codes $\textsf{GEBR}(p,\tau,k,r,q,1)$ are $(n,k)$ recoverable if
$k+r\leq p^{\nu +1}$ and are not $(n,k)$ recoverable if $k+r> p^{\nu+1}$.

We first consider that $k+r\leq p^{\nu+1}$. Let $1\leq i\leq k+r-1\leq p^{\nu+1}-1$.
In order to prove that $\textsf{GEBR}(p,\tau,k,r,q,1)$ are $(n,k)$ recoverable, we have to show
that Eq. \eqref{eq:mds-cond} has a unique solution $s(x)=0$ when $c(x)=0$ for all
$1\leq i\leq p^{\nu+1}-1$ by Lemma~\ref{lm:null-space}.

Suppose that we can find a non-zero polynomial
$s(x)\in\mathcal{C}_{p\tau}(1,\tau,q,d)$ such that
$(1+x^i)s(x)=0\bmod (1+x^{p\tau})$ holds, we can deduce a contradiction as
follows.  Without loss of generality, let $s(x)=\sum_{v=0}^{p\tau-1}s_{v}x^{v}$
and $s_{0}=1$.  Then we have
\[
s_{((\ell-1)i)\bmod p\tau}+s_{(\ell i)\bmod p\tau}=0,
\]
for $0\leq \ell\leq p\tau-1$. By induction, we have
\begin{equation}
s_{(\ell i)\bmod p\tau}=1.
\label{eq:mds-gx1-1}
\end{equation}
Let $c=\gcd(i,p\tau)$. Recall that $1\leq i\leq p^{\nu+1}-1$, we have
$c=\gcd(i,p\tau)=\gcd(i,\gamma p^{\nu+1})=\gcd(i,\gamma p^{\nu})=\gcd(i,\tau)$.
Then,
\begin{equation}
\{(\ell i)\bmod p\tau: 0\leq \ell\leq p\tau-1\}=\{\ell c: 0\leq \ell\leq \frac{p\tau}{c}-1\}.
\label{eq:mds-gx1-2}
\end{equation}
Since, in particular, $c$ divides $\tau$, we have
\begin{equation}
\{\ell \tau: 0\leq \ell\leq p-1\}\subset\{\ell c: 0\leq \ell\leq \frac{p\tau}{c}-1\}.
\label{eq:mds-gx1-3}
\end{equation}
By Eq. \eqref{eq:mds-gx1-1}, Eq. \eqref{eq:mds-gx1-2} and Eq. \eqref{eq:mds-gx1-3},
we have $s_{\ell \tau}=1$ for all $0\leq \ell\leq p-1$ and therefore
$\sum_{\ell=0}^{p-1}s_{\ell \tau}=1$, which contradicts to $\sum_{\ell=0}^{p-1}s_{\ell\tau}=0$
(since $s(x)\in\mathcal{C}_{p\tau}(1,\tau,q,d)$).

Next, we consider that $k+r> p^{\nu+1}$. We will show that Eq. \eqref{eq:mds-cond}
has a non-zero solution $s(x)\in\mathcal{C}_{p\tau}(1,\tau,q,d)$ when $c(x)=0$
and $i\in\{1,2,\ldots,k+r-1\}$.

Note that $m=p\tau=\gamma p^{\nu+1}$. Let
\begin{eqnarray}
&&s_0(x)=1+x^{p^{\nu+1}}+x^{2p^{\nu+1}}+\cdots+x^{(\gamma-1)p^{\nu+1}},\label{eq:mds-s0}\\
&&s_1(x)=x^{p^{\nu}}+x^{p^{\nu}+p^{\nu+1}}+\cdots+
x^{p^{\nu}+(\gamma-1)p^{\nu+1}},\label{eq:mds-s1}
\end{eqnarray}
and $s(x)=s_0(x)+s_1(x)\bmod (1+x^m)$. We first show that the
polynomial $s(x)=s_0(x)+s_1(x)$ is in $\mathcal{C}_{p\tau}(1,\tau,q,d)$.

Note that $s_1(x)=x^{p^\nu}s_0(x)$, we have
$1+x^m=1+x^{\gamma p^{\nu+1}}$ and $(1+x^{p^{\nu+1}})s_0(x)=1+x^{\gamma p^{\nu+1}}$.
Note that $1+x^\tau=1+x^{\gamma p^{\nu}}$ divides $1+x^{\gamma p^{\nu+1}}$.
Since $\gcd(\gamma, p)=1$, by Euclidean algorithm,
$\gcd(1+x^{p^{\nu+1}}, 1+x^{\gamma p^{\nu}})=1+x^{p^{\nu}}$.
Since  $(1+x^{p^{\nu+1}})(1+x^{p^{\nu+1}}+x^{2 p^{\nu+1}}+\cdots+ x^{(\gamma-1)p^{\nu+1}})
=1+x^{\gamma p^{\nu+1}}=(1+x^{\gamma p^{\nu}})q(x)$ and
$\gcd(1+x^{p^{\nu+1}}, 1+x^{\gamma p^{\nu}})=1+x^{p^{\nu}}$, we have
$(1+x^{\gamma p^{\nu}})|(1+x^{p^{\nu}})(1+x^{p^{\nu+1}}+x^{2 p^{\nu+1}}+\cdots+
x^{(\gamma-1)p^{\nu+1}})$. Hence,
$s(x)=s_0(x)+s_1(x)=(1+x^{p^{\nu}})(1+x^{p^{\nu+1}}+x^{2 p^{\nu+1}}+\cdots+ x^{(\gamma-1)p^{\nu+1}})
\bmod (1+x^m)$ is also divided by $1+x^{\gamma p^{\nu}}=1+x^\tau$, i.e., $s(x)\in \mathcal{C}_{p\tau}(1,\tau,q,d)$. It is clear that $s_0(x)+s_1(x)\neq 0$.

From the definitions of $s_0(x)$ and $s_1(x)$, we have
\begin{align*}
x^{p^{\nu+1}}s_0(x)=&x^{p^{\nu+1}}+x^{2 p^{\nu+1}}+x^{3p^{\nu+1}}+\cdots+x^{\gamma p^{\nu+1}}\\
=&x^{p^{\nu+1}}+x^{2p^{\nu+1}}+x^{3p^{\nu+1}}+\cdots+1\\
=&s_0(x)\bmod (1+x^m),
\end{align*}
where the second equation above comes from that $x^{\gamma p^{\nu+1}}=x^{m}=1\bmod (1+x^m)$.
Similarly, we can obtain that $x^{p^{\nu+1}}s_1(x)=s_1(x)\mod (1+x^m)$. Therefore, we have
\[
(1+x^{p^{\nu+1}})s(x)=(1+x^{p^{\nu+1}})(s_0(x)+s_1(x))=0\bmod (1+x^m),
\]
and $s(x)\neq 0$ is a solution to $(1+x^{p^{\nu+1}})s(x)=0\bmod (1+x^m)$.
The theorem is proved.
\end{proof}

We can directly obtain the following result from Theorem~\ref{thm:mds-gx1}.

\begin{corollary}
The codes $\textsf{GEBR}(p,\tau,k,r,q,1)$ with $k+r=p\tau$ are $(n,k)$ recoverable if and only if
$\tau=p^{\nu}$, where $\nu$ is a non-negative integer.
\label{thm:mds3}
\end{corollary}

By Theorem \ref{lm:mds}, the codes $\textsf{GEBR}(p,\tau,k,r,q,g(x))$ are $(n,k)$ recoverable if and only if
$\gcd(1+x^i,h(x))=1$ for all $i=1,2,\ldots,k+r-1$. When $g(x)=1$, we have
$h(x)=1+x^{\tau}+\cdots+x^{(p-1)\tau}$. According to Corollary \ref{thm:mds3},
the codes $\textsf{GEBR}(p,\tau,k,r=m-k,q,1)$ are $(n,k)$ recoverable if $p$
is an odd prime and $\tau$ is a power of $p$. Combining the results in Theorem~\ref{lm:mds}
and Corollary \ref{thm:mds3}, we can directly obtain the following theorem.
\begin{theorem}
If $p$ is an odd prime and $\tau$ is a power of $p$, then we have $\gcd(1+x^i,1+x^{\tau}+\ldots+x^{(p-1)\tau})=1$
for all $i=1,2,\ldots,p\tau-1$.
\label{thm:mds4}
\end{theorem}

Since $\mathcal{C}_{p\tau}(g(x),\tau,q,d)$ is a cyclic code,
the proposed $\textsf{GEBR}(p,\tau,k,r,q,g(x))$ can recover either a burst of up to
$\tau+\deg (g(x))$ (consecutive) erased symbols or up to $d-1$ erased symbols
in a column. Specifically, we can recover a burst of up to
$\tau+\deg (g(x))$ erased symbols as follows. First, cyclic-shift the polynomial
such that the $\tau+\deg (g(x))$ erased symbols are in the last
$\tau+\deg (g(x))$ locations. Then, encode the first $\alpha$ symbols
of the obtained shifted polynomial systematically. Finally, apply the inverse cyclic-shift to the
encoded polynomial to obtain the decoded polynomial.

\begin{example}
Consider the code $\textsf{GEBR}(3,3,6,3,2,1)$, i.e., $\tau=p=3$,
$k=6$ and $r=3$. The entries of an array in the code can be
represented by the nine polynomials
\begin{align*}
s_j(x)=&s_{0,j}+s_{1,j}x+s_{2,j}x^2+s_{3,j}x^3+s_{4,j}x^4+s_{5,j}x^5+\\
&(s_{0,j}+s_{3,j})x^6+(s_{1,j}+s_{4,j})x^7+(s_{2,j}+s_{5,j})x^8,
\end{align*}
where %$j=0,1,\ldots,8$
$0\leq j\leq 8$
and the 36 information symbols are
$s_{i,0},s_{i,1},s_{i,2},s_{i,3},s_{i,4},s_{i,5}$ with
%$i=0,1,2,3,4,5$.
$0\leq i\leq 5$.
The parity-check matrix of the code is
\[
\begin{bmatrix}
 1&1&1 & 1&1&1 & 1&1&1\\
 1 & x & x^{2}&x^{3}& x^{4}&x^{5}&x^{6}&x^{7}& x^{8}\\
 1 & x^2 & x^{4}&x^{6}& x^{8}&x^{10}&x^{12}&x^{14}& x^{16}\\
 \end{bmatrix}.
\]
According to Corollary \ref{thm:mds3}, the code
$\textsf{GEBR}(3,3,6,3,2,1)$ is $(n,k)$ recoverable.
Also, each column can recover up to three consecutive erasures.
\label{exm:gebr33}
\end{example}

\subsection{Efficient Decoding}
In the encoding/decoding procedures of $\textsf{GEBR}(p,\tau,k,r,q,g(x))$, we need to
solve a Vandermonde linear system over $\mathcal{R}_{p\tau}(q)$ such that the
solved polynomials are in $\mathcal{C}_{p\tau}(g(x),\tau,q,d)$.
In the following, we
present an efficient decoding method for solving the
Vandermonde linear system over $\mathcal{R}_{p\tau}(q)$ based on
the LU factorization of the Vandermonde matrix.
The efficient LU decoding algorithm we propose relies on an efficient algorithm for
division by $1+x^b$, and we first present the efficient division.

\subsubsection{Efficient Division by $1+x^b$ over $\mathcal{R}_{p\tau}(q)$}
We need to first give an efficient decoding algorithm for dividing by
$1+x^b$ over $\mathcal{R}_{p\tau}(q)$ before showing the efficient LU decoding method, where
$b$ is a positive integer such that $1+x^b$ and $h(x)$ are relatively
prime. Given the integer $b$ and the polynomial $f(x)\in \mathcal{C}_{p\tau}(g(x),\tau,q,d)$, we want
to solve $r(x)\in \mathcal{C}_{p\tau}(g(x),\tau,q,d)$ from the equation
\begin{equation}
(1+x^b)r(x)=f(x)\bmod (1+x^{p\tau}).
\label{eq:div}
\end{equation}
The next lemma shows a decoding algorithm for solving
$r(x)\in \mathcal{C}_{p\tau}(g(x),\tau,q,d)$ from Eq. \eqref{eq:div} when $\gcd (b,p)=1$.

\begin{lemma}
Consider Eq. \eqref{eq:div} with $1\leq b<p\tau$ and $\gcd (b,\tau)=a$.
If $\gcd (b,p)=1$, then we can first compute the coefficients $r_j$ of the polynomial $r(x)$
with $j=0,1,\ldots,a-1$ by
\begin{eqnarray}
r_j=\sum_{u=1}^{\frac{p-1}{2}}\sum_{\ell=1}^{\tau}f_{(2u-1)\tau b+\ell b+j}
 \label{eq:div1}
\end{eqnarray}
and the other coefficients of $r(x)$ iteratively by
\begin{equation}
r_{b\ell+j}=f_{b\ell+j}+r_{b(\ell-1)+j}
\label{eq:div2}
\end{equation}
with the index $\ell$ running from 1 to $p\tau/a-1$ and $j=0,1,\ldots,a-1$.
\label{lm:div}
\end{lemma}

\begin{proof}
See Appendix A.
\end{proof}

Lemma 20 in \cite{BR2019} and Lemma 4 in \cite{HOU2019} are special
case of Lemma \ref{lm:div} with $\tau=1$ and $g(x)=1$, respectively.
By Lemma \ref{lm:div}, there are
$a(\frac{p-1}{2}\cdot \tau-1)+(p\tau-a)$
XORs involved in solving $r(x)$ from Eq. \eqref{eq:div}. In particular, if $a=1$, we have that the number of involved XORs is $\frac{3p\tau-\tau-4}{2}$.

\begin{example}
Consider the example of $p=7$ and $\tau=2$. Let $f(x)=1+x+x^6+x^7+x^{10}+x^{11}+x^{12}+x^{13}\in \mathcal{C}_{7\cdot 2}((1+x^2+x^6),2,2,4)$, i.e., $f_0=f_1=f_6=f_7=f_{10}=f_{11}=f_{12}=f_{13}=1$ and $f_{2}=f_{3}=f_{4}=f_{5}=f_{8}=f_{9}=0$. We want to solve $r(x)$ from $(1+x^3)r(x)=f(x)$. According to Eq. \eqref{eq:div1}, we have
\begin{align*}
r_0=f_9+f_{12}+f_7+f_{10}+f_5+f_{8}=1.
\end{align*}
The other coefficients can be computed as
\begin{align*}
r_{1}=&r_{4}=r_6=r_9=r_{10}=r_{11}=0,\\
r_{2}=&r_3=r_{5}=r_{7}=r_{8}=r_{12}=r_{13}=1,
\end{align*}
according to Eq. \eqref{eq:div2}. Therefore, $r(x)=1+x^2+x^3+x^5+x^7+x^8+x^{12}+x^{13}$.
We can check that $r(x)\in \mathcal{C}_{7\cdot 2}((1+x^2+x^6),2,2,4)$.
\end{example}

When $\tau$ is a power of an odd prime $p$ and $g(x)=1$, $\textsf{GEBR}(p,\tau,k,r=m-k,q,1)$ are
$(n,k)$ recoverable by Corollary \ref{thm:mds3}. If $\gcd(b,p)=1$, then we can solve Eq. \eqref{eq:div} by
Lemma \ref{lm:div}; otherwise, if $b$ is a multiple of $p$, then the decoding method is as follows.
\begin{lemma}
Consider that $\tau=p^{\nu}$, where $\nu$ is a positive integer.
Let $b=up^s$, where $\gcd(u,p)=1$ and $1\leq s\leq \nu$.
We can compute the coefficients $r_{p^{\nu+1}-p^s+j}$ of the polynomial $r(x)$
in Eq. \eqref{eq:div} with $j=0,1,\ldots,p^s-1$ by
\begin{eqnarray}
r_{p^{\nu+1}-p^s+j}=\sum_{i=0}^{\frac{p^{\nu-s+1}-3}{2}}f_{2iup^s+up^s+j}
, \label{eq:div21}
\end{eqnarray}
and the other coefficients of $r(x)$ iteratively by
\[
r_{p^{\nu+1}-p^s+j+\ell up^s}=f_{p^{\nu+1}- p^s+j+\ell up^s}+r_{p^{\nu+1}- p^s+j+(\ell-1)up^s},
\]
where $\ell=1,2,\ldots,p^{\nu-s+1}-1$. Recall that the indices are taken modulo $m=p^{\nu+1}$ throughout the paper.
\label{lm:div2}
\end{lemma}
\begin{proof}
See Appendix B.
\end{proof}
By Lemma \ref{lm:div2}, there are
$p^s(\frac{p^{\nu-s+1}-3}{2})+p^{\nu+1}-p^{s}=\frac{3p\tau-5p^s}{2}$
XORs involved in solving $r(x)$ from Eq. \eqref{eq:div}.

\begin{example}
Consider the example of $p=\tau=3$. Let $f(x)=1+x+x^3+x^7\in \mathcal{C}_{3\cdot 3}(1,3,2,2)$, i.e., $f_0=f_1=f_3=f_7=1$ and $f_{2}=f_{4}=f_{5}=f_{6}=f_{8}=0$. We want to solve $r(x)$ from $(1+x^3)r(x)=f(x)$. According to Eq. \eqref{eq:div21} in
Lemma \ref{lm:div2}, we have
\begin{align*}
r_6=&f_{3}=1,\\
r_7=&f_{4}=0,\\
r_8=&f_{5}=0,
\end{align*}
and the other coefficients are $r_0=r_2=r_5=0$ and $r_1=r_3=r_4=1$.
Therefore, $r(x)=x+x^3+x^4+x^6$, which is in $\mathcal{C}_{3\cdot 3}(1,3,2,2)$.
\end{example}

\subsubsection{LU Decoding Method}
Let $\mathbf{V}_{r\times r}(\mathbf{a})$ be an $r\times r$ Vandermonde
matrix,
\begin{equation}
\mathbf{V}_{r\times r}(\mathbf{a})=\begin{bmatrix}
1 & x^{a_1}& \cdots & x^{(r-1)a_1}\\
1 & x^{a_2}& \cdots & x^{(r-1)a_2}\\
\vdots & \vdots & \ddots & \vdots\\
1 & x^{a_r}& \cdots & x^{(r-1)a_r}\\
\end{bmatrix},
\label{eq:vand-matrix}
\end{equation}
where $\mathbf{a}=(a_1,\cdots,a_r)$ and $a_1,a_2,\ldots,a_r$ are distinct integers that range from 0 to
$k+r-1$. Let $\mathbf{u}=(u_1(x),\ldots,u_r(x))\in \mathcal{C}_{p\tau}(g(x),\tau,q,d)^r$
and $\mathbf{v}=(v_1(x),\ldots,v_r(x))\in \mathcal{C}_{p\tau}(g(x),\tau,q,d)^r$.
Consider the linear equations
\begin{equation}
\mathbf{u}\mathbf{V}_{r\times r}(\mathbf{a})=\mathbf{v} \bmod (1+x^{p\tau}).
\label{eq:vand}
\end{equation}
We first review the LU factorization of a Vandermonde matrix given in \cite{YANG2005On} and then show the LU decoding algorithm for solving
$\mathbf{u}$ from the Vandermonde linear system in Eq. \eqref{eq:vand}.

\begin{theorem}\cite{YANG2005On}
For a positive integer $r$, the $r\times r$ Vandermonde matrix $\mathbf{V}_{r\times r}(\mathbf{a})$ in Eq. \eqref{eq:vand-matrix}
can be factorized into
\begin{equation*}
\mathbf{V}_{r\times r}(\mathbf{a})= \mathbf{L}_r^{(1)}\mathbf{L}_r^{(2)}\ldots \mathbf{L}_r^{(r-1)}\mathbf{U}_r^{(r-1)}\mathbf{U}_r^{(r-2)}\ldots \mathbf{U}_r^{(1)}
\label{eq7}
\end{equation*}
where $\mathbf{U}_r^{(\ell)}$ is the upper triangular matrix
\begin{figure}[H]
  \centering
  $$
\mathbf{U}_{r}^{(\ell)}=
\left[
\begin{array}{c|c}
\begin{matrix}
\mathbf{I}_{r-\ell-1}
\end{matrix}& \begin{matrix}
\mathbf{0}
\end{matrix} \\ \hline
\begin{matrix}
\mathbf{0}
\end{matrix}& \begin{matrix}
1 & x^{a_{1}}&0 &\cdots &0&0\\
0 & 1&x^{a_{2}} & \cdots &0&0\\
\vdots&\vdots&\vdots &\ddots&\vdots&\vdots\\
0 & 0&0 & \cdots &1&x^{a_{l}}\\
0 & 0&0 &\cdots &0&1\\
\end{matrix}
\end{array}
\right]
$$
  %\caption{}
  %\label{U_r^(\ell)}
\end{figure}
and $\mathbf{L}_r^{(\ell)}$ is the lower triangular matrix
\begin{figure}[H]
  \centering
%  % Requires \usepackage{graphicx}
  $$
\left[
\begin{array}{c|c}
\begin{matrix}
\mathbf{I}_{r-\ell-1}
\end{matrix}& \begin{matrix}
\mathbf{0}
\end{matrix} \\ \hline
\begin{matrix}
\mathbf{0}
\end{matrix}& \begin{matrix}
1 & 0 &\cdots &0\\
1&x^{a_{r-\ell+1}}+x^{a_{r-\ell}} & \cdots &0\\
\vdots&\vdots&\ddots &\vdots\\
0&0 &\cdots &x^{a_r}+x^{a_{r-\ell}}\\
\end{matrix}
\end{array}
\right]
$$
%  \caption{}\label{}
\end{figure}
for $\ell=1,2,\ldots,r-1$.
%, where $t=r-\ell$.
\label{thm:lu}
\end{theorem}

When $r=3$, the $3\times 3$ Vandermonde matrix $\mathbf{V}_{3\times 3}(a_1,a_2,a_3)$ can be factorized as
\begin{align*}
&\mathbf{L}_3^{(1)}\mathbf{L}_3^{(2)}\mathbf{U}_3^{(2)}\mathbf{U}_3^{(1)}\\
=&
\begin{bmatrix}
 1&0&0\\
 0&1&0\\
 0&1&x^{a_3}+x^{a_2}\\
 \end{bmatrix}\cdot
\begin{bmatrix}
 1&0&0\\
 1&x^{a_2}+x^{a_1}&0\\
 0&1&x^{a_3}+x^{a_1}\\
 \end{bmatrix}\cdot \\
&\begin{bmatrix}
 1&x^{a_1}&0\\
 0&1&x^{a_2}\\
 0&0&1\\
 \end{bmatrix}\cdot
\begin{bmatrix}
 1&0&0\\
 0&1&x^{a_1}\\
 0&0&1\\
 \end{bmatrix}.
\end{align*}

With the LU factorization of the Vandermonde matrix in Theorem \ref{thm:lu}, we
can solve the Vandermonde linear system in Eq. \eqref{eq:vand} by solving
the following linear equations
\[
\mathbf{u}\mathbf{L}_r^{(1)}\mathbf{L}_r^{(2)}\ldots \mathbf{L}_r^{(r-1)}\mathbf{U}_r^{(r-1)}\mathbf{U}_r^{(r-2)}\ldots \mathbf{U}_r^{(1)}=\mathbf{v}.
\]

%For any linear system $uV_{r\times r}(a)=v$, $a$ is a vector with $r$ terms, $a=[a_1,a_2,...,a_r]$ are the $r$ different number that determines the matrix, for an $r$-order vandermond matrix. When row vector $v$ is known, the vector $u$ can be obtained by solving linear equations.

%LU decomposition is performed by iteratively computing the Vandermonde matrix. First, let vector $Y=(y_1,y_2,...,y_r)$, so $Y*U_r^{(i)}=v\ for\;i = 1,2,...,r-1$, consider the form of upper triangular matrix $U_r^{(i)}$ in \cite{Hou2018form}, the relationship between $Y$ and $v$ as
%\begin{equation*}
%Y_j(x) = v_j(x)\ for\;j=1,2,...,r-i,
%\end{equation*}
%\begin{equation*}
%x^{a_{i+j-r}}Y_{j-1}(x)+Y_j(x) = v_j(x)\ for\;j=r-i+1,...,r,
%\label{eq3}
%\end{equation*}
%and we can obtain the row vector $Y$. So we iteratively calculate the upper triangle and the rest part can be transformed into a simplified equation that is only related to the lower triangle. Next, we denote the solved $r$ polynomials as $v'=(v_1'(x),...,v_r'(x))$, so we only need to calculate equation $u*L_r^{(i)}=v' for\ i = 1,2,...,r-1$ to obtain vector $u=(u_1(x),...,u_r(x))$. Consider the form of lower triangular matrix $L_r^{(i)}$ in \cite{Hou2018form}, the relationship between $u$ and $v$ as
%and we can obtain the row vector $u$. So far, we have solved the linear equation with a vandermonde matrix completely by LU decomposition. An algorithm1 for LU decomposition is proposed in \cite{Hou2018form}.

\begin{algorithm}[!h]
 \caption{Solving a Vandermonde Linear System}% ??¡¤?¡À¨º¨¬a
 \begin{algorithmic}[1]%¨°?DD¨°???¡À¨ºDDo?
\renewcommand{\algorithmicrequire}{ \textbf{Input:}} %Use Input in the format of Algorithm
\renewcommand{\algorithmicensure}{ \textbf{Output:}} %UseOutput in the format of Algorithm
\REQUIRE positive integer $r$, prime number $p$, integers $a_1,a_2,\ldots,a_r$, and $\mathbf{v} = (v_1(x),v_2(x),\ldots,v_r(x))\in \mathcal{C}_{p\tau}(g(x),\tau,q,d)^r$.~~\\
\ENSURE $\mathbf{u} = (u_1(x),u_2(x),\ldots,u_r(x))\in \mathcal{C}_{p\tau}(g(x),\tau,q,d)^r$ that satisfies Eq. \eqref{eq:vand}.~~\\
\renewcommand{\algorithmicensure}{ \textbf{Require:}}     % use Output in the format of Algorithm
\ENSURE $x^{a_{i_1}}+x^{a_{i_2}}$ is relatively prime to $h(x)$ over $\mathbb{F}_{q}[x]$ for all $0\le i_1 \le i_2\le k+r-1$.~~\\
\STATE
$\textbf{u}\leftarrow \textbf{v}$
\FOR{$i$ from 1 to $r-1$ }
\FOR { $j$ from $r-i+1$ to $r$}
\STATE {$u_j(x)\leftarrow u_j(x)+u_{j-1}(x)x^{a_{i+j-r}}$}\\
\ENDFOR
\ENDFOR
 \FOR {$i$ from $r-1$ down to 1 }
            \STATE {Solve $g(x)$ from $(x^{a_r}+x^{a_{r-i}})g(x) = u_r(x)$ by Lemma \ref{lm:div} or Lemma \ref{lm:div2}}\\
            \STATE {$u_r(x)\leftarrow g(x)$}\\
\FOR {$j$ from $r-1$ down to $r-i+1$}
\STATE {Solve $g(x)$ from $(x^{a_j}+x^{a_{r-i}})g(x) = (u_j(x)+u_{j+1}(x))$ by Lemma \ref{lm:div} or Lemma \ref{lm:div2}}\\
\STATE {$u_j(x)\leftarrow g(x)$}\\
\ENDFOR
\STATE {$u_{r-i}(x)\leftarrow u_{r-i}(x)+u_{r-i+1}(x)$}\\
\ENDFOR
\RETURN $\textbf{u} = (u_1(x),...,u_r(x))$\\
 \end{algorithmic}
    \label{alg:lu}
\end{algorithm}

The decoding algorithm for solving the Vandermonde linear system based on
the LU factorization of a Vandermonde matrix is given in Algorithm \ref{alg:lu}.  %\ref{al}.
In Algorithm \ref{alg:lu}, Steps 2-4 require $r(r-1)/2$ additions and
$r(r-1)/2$ multiplications and Steps 5-11 require $r(r-1)/2$ backward
additions and $r(r-1)/2$ divisions by factors of the form
$x^{a_j}+x^{a_{r-i}}$.

\begin{example}
Continue from Example \ref{exm:gebr33}, where $\tau=p=3$, $k=6$ and $r=3$. We have six information polynomials
\begin{align*}
& s_0(x)=1+x+x^3+x^4, \\
&s_1(x)=x+x^2+x^4+x^5,\\
&s_2(x)=x+x^4,\\
& s_3(x)=1+x^2+x^3+x^5, \\
&s_4(x)=x+x^2+x^7+x^8,\\
&s_5(x)=x+x^7,
\end{align*}
where each polynomial is in $\mathcal{C}_{3\cdot 3}(1,3,2,2)$. The parity-check matrix of the code is
\[
\begin{bmatrix}
 1&1&1 & 1&1&1 & 1&1&1\\
 1 & x & x^{2}&x^{3}& x^{4}&x^{5}&x^{6}&x^{7}& x^{8}\\
 1 & x^2 & x^{4}&x^{6}& x^{8}&x^{10}&x^{12}&x^{14}& x^{16}\\
 \end{bmatrix}.
\]
Therefore, we can obtain
\begin{small}
\begin{align*}
\begin{bmatrix}
 s_6(x) \text{ }s_7(x) \text{ }s_8(x)\\
 \end{bmatrix}\begin{bmatrix}
 1&x^6&x^{12}\\
 1&x^7&x^{14}\\
 1&x^8&x^{16}\\
 \end{bmatrix}
%=&\begin{bmatrix}
% s_0(x)+s_1(x)+s_2(x)+s_3(x)+s_4(x)+s_5(x) \\s_0(x)+xs_1(x) +x^2s_2(x)+x^3s_3(x)+x^4s_4(x) +x^5s_5(x) \\ s_0(x)+x^2s_1(x) +x^4s_2(x)+x^6s_3(x)+x^8s_4(x) +x^{10}s_5(x)\\
%\end{bmatrix}^T\\
=\begin{bmatrix}
 x+x^2+x^4+x^8 \\1+x+x^4+x^5+x^6+x^8 \\ 1+x^5+x^6+x^8\\
\end{bmatrix}^T.
\end{align*}
\end{small}
According to Theorem \ref{thm:lu}, the above Vandermonde matrix can be factorized as
\begin{align*}
&\begin{bmatrix}
 1&x^6&x^{12}\\
 1&x^7&x^{14}\\
 1&x^8&x^{16}\\
 \end{bmatrix}=
\begin{bmatrix}
 1&0&0\\
 0&1&0\\
 0&1&x^7+x^8\\
 \end{bmatrix}\cdot \\
&\begin{bmatrix}
 1&0&0\\
 1& x^6+x^7&0\\
 0&1&x^6+x^8\\
 \end{bmatrix}\cdot
\begin{bmatrix}
 1&x^6&0\\
 0&1&x^7\\
 0&0&1\\
 \end{bmatrix}\cdot
\begin{bmatrix}
 1&0&0\\
 0&1&x^6\\
 0&0&1\\
 \end{bmatrix}.
\end{align*}

By Algorithm \ref{alg:lu}, we can solve the three parity polynomials as follows.
First, we solve the following linear system
\begin{align*}
&\begin{bmatrix}
 s'''_6(x)&s'''_7(x)&s'''_8(x)\\
 \end{bmatrix}\begin{bmatrix}
 1&0&0\\
 0&1&x^6\\
 0&0&1\\
 \end{bmatrix}\\
=&\begin{bmatrix}
 x+x^2+x^4+x^8 \\1+x+x^4+x^5+x^6+x^8 \\ 1+x^5+x^6+x^8\\
\end{bmatrix}^T
\end{align*}
to obtain
\begin{align*}
\begin{bmatrix}
s'''_6(x)\\
s'''_7(x)\\
s'''_8(x))\\
\end{bmatrix}=\begin{bmatrix}
x +x^2 +x^4+x^8\\
1+x+ x^4+x^5+x^6+x^8\\
1+x+ x^2+x^3+x^7+x^8
\end{bmatrix}.
\end{align*}
Then, we solve the following linear system
\begin{align*}
&\begin{bmatrix}
s''_6(x) \text{ }s''_7(x) \text{ }s''_8(x)\\
 \end{bmatrix}\begin{bmatrix}
 1&x^6&0\\
 0&1&x^7\\
 0&0&1\\
 \end{bmatrix}\\
 =&\begin{bmatrix}
 x +x^2 +x^4+x^8 \\1+x+ x^4+x^5+x^6+x^8\\1+x+ x^2+x^3+x^7+x^8\\
 \end{bmatrix}^T
\end{align*}
to obtain
%\begin{footnotesize}
\begin{align*}
\begin{bmatrix}
s''_6(x)\\
s''_7(x)\\
s''_8(x))\\
\end{bmatrix}=\begin{bmatrix}
x +x^2 +x^4+x^8\\
1+x^4+x^6+x^7\\
1+x+ x^3+ x^4+ x^5+ x^8
\end{bmatrix}.
\end{align*}
%\end{footnotesize}
Next, we solve the following linear system
%\begin{footnotesize}
\begin{align*}
&[ s'_6(x) \text{ } s'_7(x) \text{ } s'_8(x)]
\begin{bmatrix}
 1&0&0\\
 1& x^6+x^7&0\\
 0&1&x^6+x^8\\
 \end{bmatrix}\\
=&\begin{bmatrix}
x +x^2 +x^4+x^8\\1+x^4+x^6+x^7\\1+x+ x^3+ x^4+ x^5+ x^8\\
\end{bmatrix}^T
\end{align*}
%\end{footnotesize}
to obtain
\begin{align*}
\begin{bmatrix}
s'_6(x)\\
s'_7(x)\\
s'_8(x)\\
\end{bmatrix}=\begin{bmatrix}
x^4 +x^5+x^7+x^8\\
x+x^2+x^5+x^7\\
x^2+ x^3+ x^5+x^6
\end{bmatrix}.
\end{align*}
Finally, we solve the following linear system
\begin{small}
\begin{align*}
[ s_6(x) \text{ } s_7(x) \text{ } s_8(x)]
\begin{bmatrix}
 1&0&0\\
 0&1&0\\
 0&1&x^7+x^8\\
 \end{bmatrix}=\begin{bmatrix}
x^4 +x^5+x^7+x^8\\x+x^2+x^5+x^7\\x^2+ x^3+ x^5+x^6\\
\end{bmatrix}^T
\end{align*}
\end{small}
to obtain
\begin{align*}
\begin{bmatrix}
s_6(x)\\ s_7(x)\\ s_8(x)\\
\end{bmatrix}=\begin{bmatrix}
x^4+x^5 +x^7+x^8\\ x+x^2+x^4+x^5\\ x^4+ x^7\\
\end{bmatrix}.
\end{align*}
Table \ref{table:example-gebr2} shows the example of $\textsf{GEBR}(3,3,6,3,2,1)$.

\begin{table}[tbh]
\caption{Example of $\textsf{GEBR}(3,3,6,3,2,1)$.}
\begin{center}
\begin{tabular}{|c|c|c|c|c|c|c|c|c|} \hline
% Column 0 &Column 1 & Column 2 & Column 3 & Column 4& Column 5& Column 6& Column 7& Column 8   \\ \hline \hline
 1  &0 &0  &1 & 0& 0&0&0&0  \\ \hline
 1  &1 &1 &0 & 1&1 &0 &1&0\\ \hline
 0  &1 &0 &1 & 1&0 &0 &1&0\\ \hline
 1  &0 &0 &1 & 0&0 &0 &0&0\\ \hline
 1  &1 &1 &0 & 0&0 &1 &1&1\\ \hline
 0  &1 &0 &1 & 0&0 &1 &1&0\\ \hline
 0  &0 &0 &0 & 0&0 &0 &0&0\\ \hline
 0  &0 &0 &0 & 1&1 &1 &0&1\\ \hline
 0  &0 &0 &0 & 1&0 &1 &0&0\\ \hline
\end{tabular}
\end{center}
\label{table:example-gebr2}
\vspace{-0.5cm}
\end{table}
\end{example}

In the following, we only evaluate the encoding complexity, and
we can obtain the decoding complexity similarly.
We define the normalized encoding complexity as the ratio of the total number of XORs involved in the encoding procedure to the number of information symbols.
In the encoding procedure of $\textsf{GEBR}(p,\tau,k,r,q,1)$, we  first compute $\tau$ parity symbols for  the first $k$ columns that takes $k\tau (p-2)$ XORs. Then, we  compute
the multiplication of $k$ polynomials and the $r\times k$ Vandermonde matrix that requires $(k-1)rp\tau$ XORs, and solve the Vandermonde linear system.
In solving the $r\times r$ Vandermonde linear system, there are $r(r-1)$ additions that require $r(r-1)p\tau$ XORs, $r(r-1)/2$ divisions that require $(r(r-1)/2) \cdot ((3p\tau-\tau-4)/2)$ XORs.\footnote{Suppose that $\gcd(b,\tau)=1$.}
Therefore, the normalized encoding complexity of $\textsf{GEBR}(p,\tau,k,r,q,1)$ is
\begin{eqnarray*}
\frac{\frac{1}4r(r-1)(7p\tau-\tau-4)+(k-1)rp\tau+k\tau(p-2)}{k(p-1)\tau}.
\end{eqnarray*}

%EBR \cite{Blaum2019} only consider the encoding complexity of $r=2$ parity columns and the encoding complexity of $r=2$ is $(3p-1)k-2$.
The encoding/decoding method of EBR is given in \cite{Blaum2019,BR2019}, and the normalized encoding complexity is
$$\frac{(\frac{1}2r^2-\frac{5}2r+2^r+rk-1)p+\frac{1}4r(r-1)(3p-5)+k(p-2)}{k(p-1)},$$
where $k=p-r$.

We give the comparison of EBR and our proposed codes about the encoding complexity in Table \ref{table:A}. The results of Table \ref{table:A} show that the proposed LU decoding method has less encoding complexity compared with the decoding methods in
\cite{Blaum2019,BR2019}.
\begin{table}[tbh]
\caption{Comparison of encoding algorithms.}
\begin{center}
\begin{tabular}{|c|c|c|c|c|c|c|} \hline
 $p$ &$\tau $ & $r$ & $k=p-r$ & EBR & GEBR &Improvement $\%$ \\ \hline \hline

 5  &1 &3 &2 & 8.88   &8.25    &7.0\\ \hline
 7  &1 &4 &3 & 13.22  &11.28   &14.7\\ \hline
  11  &1 &5 &6 & 14.42   &11.48    &20.4\\ \hline
 17  &1 &7 &10 & 25.63   &15.11    &41.0\\ \hline
19  &1 &8 &11 & 38.69  &17.67  &54.3\\ \hline
 23 &1 &10 &13 & 100.72  &22.88   &77.3\\ \hline
\end{tabular}
\end{center}
\label{table:A}
\vspace{-0.5cm}
\end{table}

%Note that the above complexity of $r = 3$ is calculated as the complexity in the case of three consecutive node failured.

\section{Generalized Expanded Independent Parity Codes}
\label{sec:geip}
In this section, we present the construction of generalized expanded independent parity (GEIP) codes.
The constructed GEIP code is denoted by $\textsf{GEIP}(p,\tau,k,r,q,g(x))$ with
encoding matrix given in Eq. \eqref{eq:matrixP2}.
\begin{equation}
\mathbf{P}_{k\times r}=
\begin{bmatrix}
 1&1&1 & \cdots & 1\\
 1 & x & x^{2}& \cdots & x^{r-1}\\
 \vdots& \vdots& \vdots & \ddots  & \vdots\\
 1& x^{k-1}& x^{2(k-1)}& \cdots & x^{(r-1)(k-1)}\\
 \end{bmatrix}.
\label{eq:matrixP2}
\end{equation}
As the first $k$ polynomials are in $\mathcal{C}_{p\tau}(g(x),\tau,q,d)$,
the computed $r$ polynomials are also in $\mathcal{C}_{p\tau}(g(x),\tau,q,d)$.
Note that the EIP code proposed in \cite{BR2019}
is a special case as $\textsf{GEIP}(p,\tau=1,k,r,q,g(x))$.

\begin{table}[tbh]
\caption{Example of $\textsf{GEIP}(p=3,\tau=3,k=3,r=2,q,g(x)=1)$, where
$s_{i,j}=s_{i-6,j}+s_{i-3,j}$ for $i=6,7,8$ and $j=0,1,2$.}
\begin{center}
\begin{tabular}{|c|c|c|c|c|} \hline
 0 &1 & 2 & 3 & 4   \\ \hline \hline
 $s_{0,0}$  &$s_{0,1}$ &$s_{0,2}$  &$s_{0,0}+s_{0,1}+s_{0,2}$ & $s_{0,0}+s_{8,1}+s_{7,2}$  \\ \hline
 $s_{1,0}$  &$s_{1,1}$ &$s_{1,2}$ &$s_{1,0}+s_{1,1}+s_{1,2}$ & $s_{1,0}+s_{0,1}+s_{8,2}$  \\ \hline
 $s_{2,0}$  &$s_{2,1}$ &$s_{2,2}$ &$s_{2,0}+s_{2,1}+s_{2,2}$ & $s_{2,0}+s_{1,1}+s_{0,2}$  \\ \hline
 $s_{3,0}$  &$s_{3,1}$ &$s_{3,2}$ &$s_{3,0}+s_{3,1}+s_{3,2}$ & $s_{3,0}+s_{2,1}+s_{1,2}$  \\ \hline
 $s_{4,0}$  &$s_{4,1}$ &$s_{4,2}$ &$s_{4,0}+s_{4,1}+s_{4,2}$ & $s_{4,0}+s_{3,1}+s_{2,2}$  \\ \hline
 $s_{5,0}$  &$s_{5,1}$ &$s_{5,2}$ &$s_{5,0}+s_{5,1}+s_{5,2}$ & $s_{5,0}+s_{4,1}+s_{3,2}$  \\ \hline
 $s_{6,0}$  &$s_{6,1}$ &$s_{6,2}$ &$s_{6,0}+s_{6,1}+s_{6,2}$ & $s_{6,0}+s_{5,1}+s_{4,2}$  \\ \hline
 $s_{7,0}$  &$s_{7,1}$ &$s_{7,2}$ &$s_{7,0}+s_{7,1}+s_{7,2}$ & $s_{7,0}+s_{6,1}+s_{5,2}$  \\ \hline
 $s_{8,0}$  &$s_{8,1}$ &$s_{8,2}$ &$s_{8,0}+s_{8,1}+s_{8,2}$ & $s_{8,0}+s_{7,1}+s_{6,2}$  \\ \hline
\end{tabular}
\end{center}
\label{table:example-geip}
\vspace{-0.5cm}
\end{table}

\begin{example}
\label{example:geip}
Consider $\alpha=6,p=\tau=k=3,r=2$ and $g(x)=1$. The 18 information
symbols are $s_{i,j}\in\mathbb{F}_{q}$ for $i=0,1,\ldots,5$ and
$j=0,1,2$. The encoding matrix of
$\textsf{GEIP}(p=3,\tau=3,k=3,r=2,q,g(x)=1)$ is
\[
\mathbf{P}_{3\times 2}=\begin{bmatrix}
1 & 1 \\
1 & x \\
1 & x^2 \\
\end{bmatrix}.
\]
Example \ref{example:geip} is illustrated in Table \ref{table:example-geip}, where
$s_{i,j}=s_{i-6,j}+s_{i-3,j}$ for $i=6,7,8$ and $j=0,1,2$.
\end{example}

\subsection{The $(n,k)$ Recoverable Property}
The codes $\textsf{GEIP}(p,\tau,k,r,q,g(x))$ are $(n,k)$ recoverable, if the
determinant of any square sub-matrix of $\mathbf{P}_{k\times r}$ in
Eq. \eqref{eq:matrixP2} is invertible in $\mathcal{C}_{p\tau}(g(x),\tau,q,d)$.
Recall that $\mathcal{C}_{p\tau}(g(x),\tau,q,d)$ is isomorphic to
$\mathbb{F}_{q}[x]/(h(x))$ by Lemma \ref{lm:isom2}, the $(n,k)$ recoverable condition
is reduced to that the determinant of any square sub-matrix of
$\mathbf{P}_{k\times r}$ in Eq. \eqref{eq:matrixP2}, after reducing
modulo $h(x)$, is invertible in $\mathbb{F}_{q}[x]/(h(x))$.

\begin{theorem}
Let $p$ be a prime number such that 2 is a primitive element in
$\mathbb{Z}_p$ and $\tau$ be a power of $p$. If $(p-1)\tau$ is larger than
\begin{equation}
\small
\frac{1}{4}kr\min(k,r)-\frac{1}{12}(\min(k,r))^3-\frac{9}{4}\max(k,r)+\frac{25}{12}\min(k,r)-4,
\label{eq:mds}
\end{equation}
then the codes $\textsf{GEIP}(p,\tau,k,r,2,g(x)=1)$ are $(n,k)$ recoverable
for $r\geq 9$.

\label{thm:mds}
\end{theorem}
\begin{proof}
When 2 is a primitive element in $\mathbb{Z}_p$ and $\tau$ be a multiple of $p$,
then $h(x)=1+x^{\tau}+x^{2\tau}+\ldots +x^{(p-1)\tau}$ is an irreducible polynomial \cite{Itoh1991Characterization}.
%We have $g(x)=1$ and $h(x)=1+x^{\tau}+x^{2\tau}+\ldots +x^{(p-1)\tau}$.
If a polynomial whose degree is less than $(p-1)\tau$, then the polynomial is relatively prime to $h(x)$.
It is sufficient to show that the maximum degree of the determinants
or all the factors of
the determinants of all square sub-matrix is less than $(p-1)\tau$. It is shown by Theorem 4 in
\cite{Hou2016On} that the maximum degree of the determinants or all the factors of
the determinants is equal to
the value on the left side in Eq.~\eqref{eq:mds} when $r\geq 9$.
Therefore, $\textsf{GEIP}(p,\tau,k,r,2,g(x)=1)$ are $(n,k)$ recoverable for $r\geq
9$, if Eq. \eqref{eq:mds} holds.
\end{proof}

If $\tau$ is a power of 2, then
\[
1+x^{\tau}+x^{2\tau}+\ldots +x^{(p-1)\tau}=(1+x+x^{2}+\ldots +x^{p-1})^{\tau}.
\]
Recall that, since $2$ is a primitive element in
$\mathbb{Z}_p$, $1+x+x^{2}+\ldots +x^{p-1}$ is irreducible over $\mathbb{F}_2$.
It is sufficient to show that the maximum degree of the determinants or all the factors of
the determinants of
all square sub-matrix is less than $p-1$, and we can obtain the following theorem
with a proof similar to that of Theorem \ref{thm:mds}.
\begin{theorem}
If 2 is a primitive element in $\mathbb{Z}_p$, $\tau$ is a power of 2 and $p-1$ is large than
the value in Eq. \eqref{eq:mds},
then the codes $\textsf{GEIP}(p,\tau,k,r,2,g(x)=1)$ are $(n,k)$ recoverable for $r\geq 9$.
\label{thm:mds2}
\end{theorem}

When $r\leq 3$, the determinant of any square sub-matrix can be
written as a multiplication of factors $1+x^i$, where $i\in
\{1,2,\ldots,k-1\}$.
Therefore,
$\textsf{GEIP}(p,\tau,k,r,q,g(x)=1)$ are $(n,k)$ recoverable for $r\leq 3$, if
$1+x^i$ and $h(x)$ are relatively prime.
When $4\leq r\leq 8$, we can list the prime numbers $p$ for which
$\textsf{GEIP}(p,\tau,k,r,q,g(x)=1)$ are $(n,k)$ recoverable with similar proof of
the MDS condition in \cite{blaum1996mds}.

Note that EIP codes share the same $(n,k)$ recoverable condition as IP codes (also
called generalized EVENODD codes \cite{Blaum1995,blaum2002evenodd} or
Blaum-Bruck-Vardy (BBV) codes \cite{Hou2016On} in the literature). By
letting $\tau$ be a power of $p$, our
$\textsf{GEIP}(p,\tau,k,r,q,g(x)=1)$ codes not only support much more
parameters, but the codes may be $(n,k)$ recoverable for some parameters with $p<k$,
compared with EIP codes.
Example \ref{example:geip} is an $(n,k)$ recoverable property code, as $1+x^i$
is relatively prime to $1+x^3+x^6$ for $i=1,2,3$.

\subsection{Encoding/Decoding Procedure}
The encoding procedure of $\textsf{GEIP}(p,\tau,k,r,q,g(x))$ is as follows.
Given $k\alpha$ information symbols $s_{i,j}$ with $i=0,1,\ldots,\alpha-1$ and $j=0,1,\ldots,k-1$,
we obtain $(m-\alpha)k$ parity symbols $s_{i,j}$ with $i=\alpha,\alpha+1,\ldots,m-1$
and $j=0,1,\ldots,k-1$ by systematically encoding such that
$s_j(x)\in \mathcal{C}_{p\tau}(g(x),\tau,q,d)$ for $j=0,1,\ldots,k-1$.
Note that when $g(x)=1$, we can do the systematic encoding by Eq. \eqref{eq:coeff};
when $g(x)\neq 1$, the systematic encoding is similar to that of
cyclic codes \cite[Ch. 7.8]{macwilliams1977}.
Then, we compute the last $r$ polynomials by Eq.~\eqref{eq:enc-matrix} with encoding
matrix $\mathbf{P}_{k\times r}$ in Eq. \eqref{eq:matrixP2}.

Next, we consider the encoding complexity of $\textsf{GEIP}(p,\tau,k,r,q,g(x)=1)$. First, we
compute $s_{(p-1)\tau+\mu,j}=\sum_{\ell=0}^{p-2}s_{\ell \tau +\mu,j}$ for $j=0,1,\ldots,k-1$ and $\mu=0,1,\ldots,\tau-1$, which takes $k\tau (p-2)$ XORs. Then, we compute the $r$ parity polynomials by Eq. \eqref{eq:enc-matrix} with encoding
matrix $\mathbf{P}_{k\times r}$ in Eq. \eqref{eq:matrixP2}, which takes $r(k-1)p\tau$ XORs.
The encoding complexity is $k\tau (p-2)+r(k-1)p\tau$ XORs. Recall that the encoding complexity of EIP code with $g(x)=1$ is $k (p-2)+r(k-1)p$ XORs. Therefore, the normalized encoding complexity of $\textsf{GEIP}(p,\tau,k,r,q,g(x)=1)$ is equal to that of EIP code with $g(x)=1$.

Suppose that $r$ information columns have failed and up to $d-1$
symbols or a burst of up to $\tau+\deg (g(x))$ symbols in each of
other $k$ columns have failed. We can first recover up to $d-1$ erased
symbols in each of the $k$ columns or a burst of up to $\tau+\deg
(g(x))$ erased symbols. Suppose that
$\textsf{GEIP}(p,\tau,k,r,q,g(x))$ is $(n,k)$ recoverable. Then, we can recover $r$
failed columns by first subtracting the other $k-r$ non-failed
information polynomials from each of $r$ parity polynomials and then
solving the $r\times r$ Vandermonde linear equations by applying the
LU decoding method in Algorithm \ref{alg:lu}.
Note that if some of the erased columns are among the $r$ parity
columns, we can not formulate the $r\times r$ Vandermonde linear
equations and the LU decoding method is not applicable.

\begin{example}
Suppose that columns 0 and 1 in Example \ref{example:geip} have failed.
%and symbols
%\[
%s_{0,2},s_{1,2},s_{2,2},s_{0,3},s_{1,3},s_{2,3},s_{0,4},s_{1,4},s_{2,4},
%\]
%are erased.
%We can first recover the erased symbols in the last three columns by
%\begin{align*}
%s_{0,2}=&s_{6,2}+s_{3,2},\\
%s_{1,2}=&s_{7,2}+s_{4,2},\\
%s_{2,2}=&s_{8,2}+s_{5,2},\\
%s_{0,3}=&s_{6,3}+s_{3,3},\\
%s_{1,3}=&s_{7,3}+s_{4,3},\\
%s_{2,3}=&s_{8,3}+s_{5,3},\\
%s_{0,4}=&s_{6,4}+s_{3,4},\\
%s_{1,4}=&s_{7,4}+s_{4,4},\\
%s_{2,4}=&s_{8,4}+s_{5,4}.
%\end{align*}
We can obtain three polynomials
\[
s_j(x)=s_{0,j}+s_{1,j}x+\cdots+s_{8,j}x^8,
\]
for $j=2,3,4$. By subtracting $s_2(x)$ from each of $s_3(x)$ and $s_4(x)$, we have
\begin{align*}
\begin{bmatrix}
s_3(x)-s_2(x) & s_4(x)-x^2s_2(x)
\end{bmatrix}=
\begin{bmatrix}
s_0(x) \text{ } s_1(x)
\end{bmatrix}
\begin{bmatrix}
1 & 1\\
1 & x\\
\end{bmatrix}.
\end{align*}
Therefore, we can solve $s_0(x)$ as
\[
s_0(x)=\frac{s_4(x)-x^2s_2(x)-x(s_3(x)-s_2(x))}{1-x},
\]
by Lemma \ref{lm:div} and $s_1(x)=s_3(x)-s_2(x)-s_0(x)$.
\end{example}

\section{Minimum Symbol Distance}
\label{sec:distance}
In this section, we consider the symbol distance of the proposed array codes, which
is the number of symbols in which two codewords differ. Minimum
symbol distance is a measure of the maximum number of failed symbols
that the codes can tolerate.

\begin{theorem}
Let $D$ be the minimum symbol distance of an $(n,k)$ recoverable property array codes
constructed by the coding method in Section \ref{sec:framework}.
Then, we have $D\geq d(r+1)$.
\label{thm:lower-dist}
\end{theorem}
\begin{proof}
The proof is similar to the one of Lemma 28 in \cite{BR2019}. For
completeness, we present the proof. Since the code is $(n,k)$ recoverable, there are
at least $r+1$ non-zero columns, and since each
non-zero column has weight at least $d$, we obtain $D\geq d(r+1)$.
\end{proof}

We first consider the $(n,k)$ recoverable property array codes with each entry of the encoding
matrix $\mathbf{P}_{k\times r}$ being a power of $x$.

\begin{theorem}
If each entry of $\mathbf{P}_{k\times r}$ is a power of $x$, then $D= d(r+1)$.
\label{thm:dist1}
\end{theorem}
\begin{proof}
Let $k-1$ out of the $k$ data polynomials be zero and the remaining data polynomial
be a non-zero polynomial in $\mathcal{C}_{p\tau}(g(x),\tau,q,d)$ with weight $d$.
Note that the multiplication of $x^i$ and a polynomial can be implemented by
cyclic-shifting $i$ positions of the polynomial. By encoding the $k$ data polynomials,
we have that the obtained $r$ parity polynomials are all in $\mathcal{C}_{p\tau}(g(x),\tau,q,d)$
with weight $d$. Therefore, we obtain a code with symbol distance being $d(r+1)$ and we can obtain the result
by Theorem \ref{thm:lower-dist}.
\end{proof}
By Theorem \ref{thm:dist1}, the minimum symbol distance of $\textsf{GEIP}(p,\tau,k,r,q,g(x))$ is $d(r+1)$.
Next, we consider the minimum symbol distance of $\textsf{GEBR}(p,\tau,k,r,q,1)$.

\begin{lemma}
Let $s(x)=x^i(1+x^{j\tau})\in \mathcal{C}_{p\tau}(1,\tau,q,2)$ with weight 2, where $i,j$ are
integers with $0\leq i<p\tau$ and $1\leq j<p\neq 2$. If
$c(x)\in \mathcal{C}_{p\tau}(1,\tau,q,d)$ such that $(1+x^a)c(x)=(1+x^b)s(x)\bmod 1+x^{p\tau}$,
where $0<a,b<\tau$ and $a\neq b$,
then the weight of $c(x)$ is larger than 2.
\label{lm:weight}
\end{lemma}
\begin{proof}
Note that $d>1$ because $(1+x^\tau)|f(x)$, where $f(x)$ is the generator polynomial of $\mathcal{C}_{p\tau}(1,\tau,q,d)$. We assume that the weight of $c(x)$ is 2, we can always
obtain a contradiction as follows and therefore the weight of
$c(x)$ is larger than 2. Let $c(x)=x^{\ell}(1+x^c)$
with $0\leq \ell <p\tau$. Since $c(x)\in \mathcal{C}_{p\tau}(1,\tau,q,d)$,
$c\neq 0$ is a multiple of $\tau$. From $(1+x^a)c(x)=(1+x^b)s(x)$, we have
\[
x^{\ell-i}(1+x^a+x^c+x^{a+c})=1+x^{b}+x^{j\tau}+x^{b+j\tau}.
\]
Let $e=\ell-i$. Then $0\leq e<p\tau$. When $e=0$, we have
$\{0,a,c,a+c\}=\{0,b,j\tau,b+j\tau\}\bmod p\tau$. By the assumption,
$a\neq b$, $a\neq 0$ and $a\neq j\tau$. Thus, $a=b+j\tau\bmod p\tau$. Since $b<\tau$ and $j<p$, we have $a=b+j\tau$ which is impossible due to $a<\tau$ and $j\ge 1$.
%Since $c$
%is a multiple of $\tau$, then $c=j\tau$ and $a+c=b$. We can obtain
%that $j=\frac{p}{2}$, which contradicts with the condition that $j$
%is an integer number and $p$ is an odd prime number.

Next, we consider that $e\neq 0$ and we have $\{e,e+a,e+c,e+a+c\}=\{0,b,j\tau,b+j\tau\}\bmod p\tau$. Note that $\ell+c<p\tau$. Since $c\neq 0$ is a multiple of $\tau$, we have  $\ell<(p-1)\tau$.
Assume that $e=0\bmod p\tau$. This is impossible since $e<p\tau$.
Assume that $e+a=0\bmod p\tau$, i.e., $\ell-i+a=p\tau$. That is,
$a-i=p\tau-\ell>\tau$ which contradicts to the fact $a<\tau$.  Assume
that $e+c=0\bmod p\tau$, i.e., $\ell-i+c=p\tau$. Since
$\ell+c<p\tau$, it is impossible. Finally, assume that $e+a+c=0\bmod
p\tau$. We have $\ell-i+a+c=p\tau$ due to $\ell+c<p\tau$ and
$a-i<\tau$.  Since $b<j\tau<b+j\tau$ and and $e<e+a<e+c$, $e=b<\tau$,
$e+a=j\tau<2\tau$, and $e+c=b+j\tau$. Hence, $j=1$ and $c=\tau$.
Therefore, $a+b=p\tau-\tau=(p-1)\tau$ and $e+a=b+a=j\tau=\tau$. Since
$p>2$, we have a contradiction.
Therefore, the weight of $c(x)$ is larger than 2.
\end{proof}

\begin{theorem}
Suppose the codes $\textsf{GEBR}(p,\tau,k,r,q,1)$ are $(n,k)$ recoverable. When $r=2$ and $\tau\leq \lfloor \frac{k+1}{2}\rfloor$, the minimum symbol distance of $\textsf{GEBR}(p,\tau,k,r,q,1)$ is $2(r+1)=6$. When $r=2$ and $\tau> k+1$, the minimum symbol distance of $\textsf{GEBR}(p,\tau,k,r,q,1)$ is 8.
When $r=3$, and $\tau \leq \lfloor \frac{k+2}{3}\rfloor$, the minimum symbol distance of $\textsf{GEBR}(p,\tau,k,r,q,1)$ is $2(r+1)=8$.
\label{thm:dist2}
\end{theorem}
\begin{proof}
By Theorem \ref{thm:lower-dist}, if we can find a codeword composed of $r+1$ non-zero polynomials each with weight 2 and $k-1$ zero polynomials, then the minimum symbol distance is $2(r+1)$.

When $r=2$, by Theorem~\ref{thm:lower-dist}, each non-zero codeword contains at least three non-zero polynomials.
Without loss of generality, suppose that the three non-zero polynomials are $s_{\alpha}(x),s_{\beta}(x),s_{\gamma}(x)$ and the other $k-1$ polynomials are zero, where $0\leq \alpha<\beta<\gamma\leq k+1$. Suppose that the weight of $s_{\alpha}(x)$ is 2.
According to Eq. \eqref{eq:matrixH2},
we obtain that
\begin{align*}
\begin{bmatrix}
s_{\alpha}(x)\\
x^{\alpha}s_{\alpha}(x)\\
\end{bmatrix}=
\begin{bmatrix}
1 & 1 \\
x^{\beta} & x^{\gamma} \\
\end{bmatrix}\cdot
\begin{bmatrix}
s_{\beta}(x)\\
s_{\gamma}(x)\\
\end{bmatrix}.
\end{align*}
Therefore, we can compute that $(x^{\beta}+x^{\gamma})s_{\gamma}(x)=(x^{\alpha}+x^{\beta})s_{\alpha}(x)$ and $(x^{\beta}+x^{\gamma})s_{\beta}(x)=(x^{\alpha}+x^{\gamma})s_{\alpha}(x)$.
If $\lfloor \frac{k+1}{2}\rfloor \geq \tau$, we have $k+1\geq 2\tau$. Let $(\alpha,\beta,\gamma)=(0,\tau,2\tau)$ and $s_{0}(x)=1+x^{\tau}\in \mathcal{C}_{p\tau}(1,\tau,q,2)$ with weight $2$, then we can obtain $s_{\tau}(x)=x^{\tau}+x^{p\tau-\tau}$ and $s_{2\tau}(x)=1+x^{p\tau-\tau}$, which are both with weight $2$.
%When the weight of $s_{\alpha}(x)$ is larger than $2$, the total weight of the codeword is at least $7$ due to the fact that every non-zero column has weight at least $2$.
Therefore, the minimum symbol distance of $\textsf{GEBR}(p,\lfloor \frac{k+1}{2}\rfloor \geq \tau,k,r=2,q,1)$ is $2(r+1)=6$.

If $k+1 < \tau$, we have $0<\gamma-\alpha\leq k+1<\tau$ and $0<\gamma-\beta\leq k<\tau$.
Suppose that the weight of $s_{\alpha}(x)$ is 2, by Lemma \ref{lm:weight}, the weight of $s_{\beta}$
is larger than 2. Since $s_{\beta}(x)\in \mathcal{C}_{p\tau}(1,\tau,q,2)$ (the number of non-zero coefficients
of $s_{\beta}(x)$ is an even number by Lemma \ref{lm:coeff1}) and the weight of
$s_{\beta}(x)$ is larger than 2, the weight of $s_{\beta}(x)$ is no less than~4.
Consider the polynomial $s_{\gamma}(x)$. Since the weight of  $s_{\gamma}(x)$ is at least $2$,
the minimum symbol distance of $\textsf{GEBR}(p,\tau>k+1,k,r=2,q,1)$ is at least 8.
Let $(\alpha,\beta,\gamma)=(0,\beta,2\beta)$ and $s_{0}=1+x^{\tau}$, where $\beta\le (k+1)/2<\tau$. Then we obtain
$s_{2\beta}(x)=x^{p\tau-\beta}+x^{(p\tau-\beta+\tau)\bmod p\tau}$, which has  weight 2. Since $1+x^{2\beta}=(1+x^\beta)^2$, we can obtain
$s_{\beta}(x)=x^{p\tau-\beta}(1+x^{\tau})(1+x^\beta)$, i.e., $s_{\beta}(x)=1+x^\tau+x^{p\tau-\beta}+x^{(p\tau-\beta+\tau)\bmod p\tau}$, which is with weight 4. We can thus obtain
that the minimum symbol distance of $\textsf{GEBR}(p,\tau>k+1,k,r=2,q,1)$ is 8.

When $r=3$, by Theorem~\ref{thm:lower-dist}, we have at least four non-zero polynomials.
Without loss of generality, suppose that the four non-zero polynomials are $s_{\alpha}(x),s_{\beta}(x),s_{\gamma}(x),s_{\eta}(x)$ and the other $k-1$ polynomials are zero, where $0\leq \alpha<\beta<\gamma<\eta \leq k+2$. We assume that the weight of $s_{\alpha}(x)$ is 2. By Eq.~\eqref{eq:matrixH2},
we have
\begin{align*}
\begin{bmatrix}
s_{\alpha}(x)\\
x^{\alpha}s_{\alpha}(x)\\
x^{2\alpha}s_{\alpha}(x)\\
\end{bmatrix}=
\begin{bmatrix}
1 & 1 & 1 \\
x^{\beta} & x^{\gamma} & x^{\eta} \\
x^{2\beta} & x^{2\gamma} & x^{2\eta} \\
\end{bmatrix}\cdot
\begin{bmatrix}
s_{\beta}(x)\\
s_{\gamma}(x)\\
s_{\eta}(x)\\
\end{bmatrix}.
\end{align*}

%{\bf Edit up to here}

{Since $\lfloor \frac{k+2}{3}\rfloor\geq \tau$, we have $k+2\geq 3\tau$. Let $(\alpha,\beta,\gamma,\eta)=(0,\tau,2\tau,3\tau)$ and $s_{0}(x)=1+x^{\tau}$, then
we can compute $s_{\tau}(x),s_{2\tau}(x),s_{3\tau}(x)$ as follows,
\begin{align*}
\begin{bmatrix}
s_{\tau}(x)\\
s_{2\tau}(x)\\
s_{3\tau}(x)\\
\end{bmatrix}=&\begin{bmatrix}
x^{\tau}+x^{p\tau -2\tau}\\
1+x^{p\tau -3\tau}\\
x^{p\tau-3\tau}+x^{p\tau-2\tau}\\
\end{bmatrix},
\end{align*}
which have all weight 2.}
%When the weight of $s_{\alpha}(x)$ is larger than 2, the total weight of the codeword is at least $9$ due to the fact that every non-zero column has weight at least 2.
Therefore, the minimum symbol distance is $2(r+1)=8$ when $\lfloor \frac{k+2}{3}\rfloor\geq \tau$ and $r=3$.
\end{proof}

By Theorem \ref{thm:dist2}, the minimum symbol distance of $\textsf{GEBR}(p,\tau>k+1,k,r=2,q,1)$ is larger than that of $\textsf{GEBR}(p,\tau=1,k,r=2,q,1)$, i.e., EBR codes with $r=2$.

\begin{theorem}
Suppose the codes $\textsf{GEBR}(p,\tau,k,r,q,1)$ are $(n,k)$ recoverable. When $r= 4$, $g(x)=1$ and $\tau\leq \lfloor \frac{k+3}{4}\rfloor$, the minimum symbol distance of $\textsf{GEBR}(p,\tau,k,4,q,1)$ is no larger than 12.
%When $r= 5$, $p>5$, $g(x)=1$ and $\tau\leq \lfloor \frac{k+4}{5}\rfloor$, the minimum symbol distance of $\textsf{GEBR}(p,\tau,k,5,q,1)$ is no larger than 18.
\label{thm:dist3}
\end{theorem}
\begin{proof}
It is sufficient to find a codeword such that the symbol distance is 12 when $r=4$. As the code $\textsf{GEBR}(p,\tau,k,r,q,1)$ is $(n,k)$ recoverable, each codeword contains at least $r+1$ non-zero polynomials that are in $\mathcal{C}_{p\tau}(1,\tau,q,d)$, with the weight $d$ of each non-zero polynomial being a multiple of 2 by Lemma \ref{lm:coeff1}.

Consider $r=4$. Without loss of generality, suppose that the five non-zero polynomials are $s_{\alpha}(x),s_{\beta}(x),s_{\gamma}(x),s_{\delta}(x),s_{\eta}(x)$ and the other $k-1$ polynomials are zero, where $0\leq \alpha<\beta<\gamma<\delta<\eta\leq k+3$. Suppose that the weight of $s_{\alpha}(x)$ is 2.
According to Eq. \eqref{eq:matrixH2},
we have
\begin{align*}
&\begin{bmatrix}
s_{\alpha}(x) &
x^{\alpha}s_{\alpha}(x) &
x^{2\alpha}s_{\alpha}(x) &
x^{3\alpha}s_{\alpha}(x)\\
\end{bmatrix}\\
=&\begin{bmatrix}
s_{\beta}(x) &
s_{\gamma}(x) &
s_{\delta}(x)&
s_{\eta}(x)\\
\end{bmatrix}\cdot
\begin{bmatrix}
1 & x^{\beta} & x^{2\beta} & x^{3\beta}\\
1 & x^{\gamma} & x^{2\gamma} & x^{3\gamma}\\
1 & x^{\delta} & x^{2\delta}  & x^{3\delta}\\
1 & x^{\eta} & x^{2\eta}  & x^{3\eta}\\
\end{bmatrix}.
\end{align*}
Since $\tau\leq \lfloor \frac{k+3}{4}\rfloor$,
we have $4\tau\leq k+3$.
Let $(\alpha,\beta,\gamma,\delta,\eta)=(0,\tau,2\tau,3\tau,4\tau)$ and $s_0(x)=1+x^{\tau}$, then we can take
\begin{align*}
s_{\tau}(x)=&x^{\tau}+x^{(p-3)\tau},\\
s_{2\tau}(x)=&x^{(p-5)\tau}+x^{(p-3)\tau}+x^{(p-2)\tau}+1,\\
s_{3\tau}(x)=&x^{(p-6)\tau}+x^{(p-2)\tau},\\
s_{4\tau}(x)=&x^{(p-6)\tau}+x^{(p-5)\tau}.
\end{align*}

Therefore, the minimum symbol distance is no larger than 12 when $r=4$ and $\tau\leq \lfloor \frac{k+3}{4}\rfloor$.
\end{proof}

Lemma 30 in \cite{BR2019} is a special case of Theorem \ref{thm:dist2} with $\tau=1$ and $2\leq r\leq 3$.
To determine the minimum symbol distance of other parameters is an open problem.

\section{Recovery of Erased Lines of Slope $i$ in $\textsf{GEBR}(p,\tau,k,r,q,1)$ Codes}
\label{sec:lines}
In this section, we assume $\textsf{GEBR}(p,\tau,k,r,q,1)$  are $(n,k)$ recoverable and want to show that
$\textsf{GEBR}(p,\tau,k,r,q,1)$ can recover some erased
lines of slope $i$ with $i=0,1,\ldots,r-1$, under some constraint.
We first consider the code $\textsf{GEBR}(p,1,k,r,q,1)$ with
$\tau=1$, that is, an EBR code. Recall that the $k+r$ symbols in line $\ell$ of slope $i$ of the $m \times (k+r)$ codeword array are $s_{\ell,0}, s_{\ell-i,1},s_{\ell-2i,2},\cdots,s_{\ell-(k+r-1)i,k+r-1}$ for $\ell=0,1,\ldots,m-1$ and $i=0,1,\ldots,r-1$.

\begin{theorem}
%If $1+x^{i}$ is relatively prime to $1+x+\ldots+x^{p-1}$ for $i=1,2,\ldots,p-1$, then
The code $\textsf{GEBR}(p,1,k,r,q,1)$
can recover any $r$ erased lines $e_1,e_2,\ldots,e_r$ of slope $i$ for $0\leq i\leq r-1$, if and only if
the following matrix
\[
\begin{bmatrix}
1 & 1 & \cdots & 1\\
x^{e_1} & x^{e_2} & \cdots & x^{e_r}\\
x^{e_1\cdot 2^{-1}} & x^{e_2\cdot 2^{-1}} & \cdots & x^{e_r\cdot 2^{-1}}\\
\vdots & \vdots & \ddots & \vdots \\
x^{e_1\cdot(r-i-1)^{-1}} & x^{e_2\cdot(r-i-1)^{-1}} & \cdots & x^{e_r\cdot(r-i-1)^{-1}}\\
x^{e_1\cdot (p-1)^{-1}} & x^{e_2\cdot (p-1)^{-1}} & \cdots & x^{e_r\cdot (p-1)^{-1}}\\
x^{e_1\cdot (p-2)^{-1}} & x^{e_2\cdot (p-2)^{-1}} & \cdots & x^{e_r\cdot (p-2)^{-1}}\\
\vdots & \vdots & \ddots & \vdots \\
x^{e_1\cdot (p-i)^{-1}} & x^{e_2\cdot (p-i)^{-1}} & \cdots & x^{e_r\cdot (p-i)^{-1}}\\
\end{bmatrix}
\]
is invertible over $\mathbb{F}_{q}[x]/(1+x+\ldots+x^{p-1})$,
where $0\leq e_1 <e_2<\cdots <e_r\leq p-1$ and $k+r\leq p$. Note that for any integer $\ell$ with
$1\leq \ell\leq p-1$, $\ell^{-1}$ is the inverse of $\ell$ in $\mathbb{Z}_p$, i.e.,
$\ell^{-1}\ell=1\bmod p$.
\label{thm:reco-tau1}
\end{theorem}
\begin{proof}
When $\tau=1$, we have $p\geq k+r$. Suppose that $r$ lines $e_1,e_2,\ldots,e_r$ of slope 0 are erased, where $0\leq e_1 <e_2<\cdots <e_r\leq p-1$.
For $\ell=0,1,\ldots,p-1$, we represent $k+r$ symbols $s_{\ell,0},s_{\ell,1},\ldots,s_{\ell,k+r-1}$ by the polynomial
\[
\bar{s}_{\ell}(x)=s_{\ell,0}+s_{\ell,1}x+\ldots+s_{\ell,k+r-1}x^{k+r-1}
\]
over $\mathbb{F}_{q}[x]/(1+x^{p})$. As
\[
s_{\ell,0}+s_{\ell,1}+\ldots+s_{\ell,k+r-1}=0,
\]
by Lemma~\ref{lm:coeff1}, we have\footnote{We can set $s_{\ell, k+r}, s_{\ell,k+r+1}, \ldots, s_{\ell, p-1}$ all zeros.}
\[
\bar{s}_\ell(x)\in \mathcal{C}_{p}(1,1,q,d).
\]
$\mathcal{C}_{p}(1,1,q,d)$ is an ideal of $\mathbb{F}_{q}[x]/(1+x^{p})$ and is isomorphic to $\mathbb{F}_{q}[x]/(1+x+\ldots+x^{p-1})$ by Lemma~\ref{lm:isom2}. For $j=0,1,\ldots,k+r-1$, the summation of the $p$ symbols in column $j$ is zero, we can thus compute
\[
s_{e_1,j}+s_{e_2,j}+\cdots+s_{e_r,j}=\sum_{\ell=1,\ell\neq e_2-e_1,\ldots,e_r-e_1}^{p-1}s_{\ell +e_1,j},
\]
i.e.,
\begin{equation}
\bar{s}_{e_1}(x)+\bar{s}_{e_2}(x)+\cdots+\bar{s}_{e_r}(x)=\sum_{\ell=1,\ell\neq e_2-e_1,\ldots,e_r-e_1}^{p-1}\bar{s}_{\ell +e_1}(x).
\label{eq:tau1-reco}
\end{equation}
According to row $\ell$ of Eq. \eqref{eq:matrixH2}, we have
\begin{align*}
&0=s_0(x)+x^{\ell}s_1(x)+\cdots+x^{(p-1)\ell}s_{p-1}(x)\\
&=(s_{0,0}+s_{1,0}x+\cdots+s_{p-1,0}x^{p-1})+\\
&x^{\ell}(s_{0,1}+s_{1,1}x+\cdots+s_{p-1,1}x^{p-1})+\\
&\cdots+x^{(p-1)\ell}(s_{0,p-1}+s_{1,p-1}x+\cdots+s_{p-1,p-1}x^{p-1})\\
&=(s_{0,0}+s_{0,1}x^{\ell}+\cdots+s_{0,p-1}x^{(p-1)\ell})+\\
&x(s_{1,0}+s_{1,1}x^{\ell}+\cdots+s_{1,p-1}x^{(p-1)\ell})+\\
&\cdots+x^{p-1}(s_{p-1,0}+s_{p-1,1}x^{\ell}+\cdots+s_{p-1,p-1}x^{(p-1)\ell})\\
&=(s_{0,0}+s_{1,p-\ell^{-1}}+\cdots+s_{p-1,p-(p-1)\ell^{-1}})+\\
&x^{\ell}(s_{0,1}+s_{1,p-\ell^{-1}+1}+\cdots+s_{p-1,p-(p-1)\ell^{-1}+1})+\cdots+\\
&x^{(p-1)\ell}(s_{0,p-1}+s_{1,p-\ell^{-1}+p-1}+\cdots+s_{p-1,p-(p-1)\ell^{-1}+p-1}),
\end{align*}
where $\ell^{-1}\ell=1\bmod p$. Note that $\ell\cdot i\neq \ell \cdot j\bmod p$ for
$i\neq j\in \{0,1,\ldots,p-1\}$ and $\{0,\ell,\cdots,(p-1)\ell\}=\{0,1,\cdots,p-1\}\bmod p$,
we have
\[
s_{0,j}+s_{1,p-\ell^{-1}+j}+\cdots+s_{p-1,p-(p-1)\ell^{-1}+j}=0
\]
for $j=0,1,\ldots,p-1$.
Therefore, we can obtain that
\[
\bar{s}_{0}(x)+x^{\ell^{-1}}\bar{s}_{1}(x)+x^{2\ell^{-1}}\bar{s}_{2}(x)+\cdots+x^{(p-1)\ell^{-1}}\bar{s}_{p-1}(x)=0,
\]
where $\ell=1,2,\ldots,r-1$. From the above equation, we then can compute
\begin{eqnarray}
&&x^{e_1 \ell^{-1}}\bar{s}_{e_1}(x)+x^{e_2 \ell^{-1}}\bar{s}_{e_2}(x)+\cdots+x^{e_r \ell^{-1}}\bar{s}_{e_r}(x)\nonumber\\
&=&\sum_{u=1,u\neq e_2-e_1,\ldots,e_r-e_1}^{p-1}x^{(u+e_1) \ell^{-1}}\bar{s}_{u+e_1}(x).
\label{eq:tau1-reco1}
\end{eqnarray}
By Eq. \eqref{eq:tau1-reco} and Eq. \eqref{eq:tau1-reco1}, we can obtain
\begin{align}
&\begin{bmatrix}
1 & 1 & \cdots & 1\\
x^{e_1} & x^{e_2} & \cdots & x^{e_r}\\
\vdots & \vdots & \ddots & \vdots \\
x^{e_1(r-1)^{-1}} & x^{e_2(r-1)^{-1}} & \cdots & x^{e_r(r-1)^{-1}}\\
\end{bmatrix} \begin{bmatrix}
\bar{s}_{e_1}(x)\\ \bar{s}_{e_2}(x)\\ \vdots \\ \bar{s}_{e_r}(x)\\
\end{bmatrix}\nonumber\\
=&\begin{bmatrix}
\sum_{u=1,u\neq e_2-e_1,\ldots,e_r-e_1}^{p-1}\bar{s}_{u +e_1}(x)\\
\sum_{u=1,u\neq e_2-e_1,\ldots,e_r-e_1}^{p-1}x^{u+e_1}\bar{s}_{u+e_1}(x)\\
\vdots \\
\sum_{u=1,u\neq e_2-e_1,\ldots,e_r-e_1}^{p-1}x^{(u+e_1) (r-1)^{-1}}\bar{s}_{u+e_1}(x)\\
\end{bmatrix}.
\label{eq: invertible V}
\end{align}
By applying the isomorphism $\theta: \mathcal{C}_{p}(1,1,q,d) \rightarrow \mathbb{F}_{q}[x]/(1+x+\ldots+x^{p-1})$ in Lemma \ref{lm:isom2} for the above linear equations,
we can show that, if the determinant of the above $r\times r$ matrix is invertible
over $\mathbb{F}_{q}[x]/(1+x+\ldots+x^{p-1})$, then we can
compute $\bar{s}_{e_1}(x),\bar{s}_{e_2}(x),\ldots,\bar{s}_{e_r}(x)$ by first solving
the linear equations over $\mathbb{F}_{q}[x]/(1+x+\ldots+x^{p-1})$ and then
applying the inverse isomorphism $\theta^{-1}$.

Next, we consider that $r$ lines $e_1,e_2,\ldots,e_r$ of slope $i$ are erased, where $0\leq e_1 <e_2<\cdots <e_r\leq p-1$ and $i=1,2,\ldots,r-1$. For $\ell=0,1,\ldots,p-1$, we represent the $k+r$ symbols
$s_{\ell,0},s_{\ell-i,1},s_{\ell-2i,2},\ldots,s_{\ell-i(k+r-1),k+r-1}$
in the line of slope $i$ by the polynomial
\[
\bar{s}_{\ell}(x)=s_{\ell,0}+s_{\ell-i,1}x+\ldots+s_{\ell-i(k+r-1),k+r-1}x^{k+r-1},
\]
which is in $\mathcal{C}_{p}(1,1,q,d)$, as the summation of the $k+r$ symbols is zero.
{According to row $i+1$ of Eq. \eqref{eq:matrixH2},
we have that the summation of the $k+r$ symbols in the line of slope
$i+1$ is zero, i.e.,
\begin{equation}
s_{j,0}+s_{j-(i+1),1}+s_{j-2(i+1),2}+\cdots+s_{j-(k+r-1)(i+1),k+r-1}=0
\label{eq:thm-rec-sum0}
\end{equation}
for $j=0,1,\ldots,p-1$. Note that all the indices are taken modulo $p$ in the proof
and $s_{j,k+r},s_{j,k+r+1},$ $\cdots,s_{j,p-1}$ are all zero.
Then, we can obtain that
\begin{align*}
&\bar{s}_{0}(x)+x\bar{s}_{1}(x)+x^{2}\bar{s}_{2}(x)+\cdots+
x^{p-1}\bar{s}_{p-1}(x)\\
=&\sum_{j=0}^{p-1}s_{-ij,j}x^j+x(\sum_{j=0}^{p-1}s_{1-ij,j}x^j)+
\cdots+x^{p-1}(\sum_{j=0}^{p-1}s_{p-1-ij,j}x^j)\\
=&(s_{0,0}+s_{p-1-i,1}+s_{p-2-2i,2}+\cdots+s_{1-(p-1)i,p-1})+\\
&(s_{-i,1}+s_{1,0}+s_{2-(p-1)i,p-1}+\cdots+s_{p-1-2i,2})x\\
&+(s_{-2i,2}+s_{1-i,1}+s_{2,0}+\cdots+s_{p-1-3i,3})x^2+\cdots+\\
&(s_{-i(p-1),p-1}+s_{1-i(p-2),p-2}+\cdots+s_{p-1,0})x^{p-1}\\
=&0,
\end{align*}
where the last equation comes from Eq. \eqref{eq:thm-rec-sum0}.
Similarly, according to row $i+\ell$ of Eq. \eqref{eq:matrixH2}, we can compute that
\begin{align*}
\bar{s}_{0}(x)+x^{\ell^{-1}}\bar{s}_{1}(x)+\cdots+x^{(p-1)\ell^{-1}}\bar{s}_{p-1}(x)=0,
\end{align*}
where $\ell=1,2,\ldots,r-i-1, -1,-2,\ldots,-i$, then we can compute
\begin{eqnarray}
&&x^{e_1 \ell^{-1}}\bar{s}_{e_1}(x)+x^{e_2 \ell^{-1}}\bar{s}_{e_2}(x)+\cdots+x^{e_r \ell^{-1}}\bar{s}_{e_r}(x)\nonumber \\
&=&\sum_{u=1,u\neq e_2-e_1,\ldots,e_r-e_1}^{p-1}x^{(u+e_1) \ell^{-1}}\bar{s}_{u+e_1}(x).
\label{eq:tau1-reco2}
\end{eqnarray}
}

Since the summation of the $p$ symbols in each column is zero, we have
\[
\bar{s}_0(x)+\bar{s}_1(x)+\ldots+\bar{s}_{p-1}(x)=0,
\]
together with Eq. \eqref{eq:tau1-reco2}, we can obtain
\begin{small}
\begin{align}
&\begin{bmatrix}
1 & 1 & \cdots & 1\\
x^{e_1} & x^{e_2} & \cdots & x^{e_r}\\
\vdots & \vdots & \ddots & \vdots \\
x^{e_1(r-i-1)^{-1}} & x^{e_2(r-i-1)^{-1}} & \cdots & x^{e_r(r-i-1)^{-1}}\\
x^{e_1(p-1)^{-1}} & x^{e_2(p-1)^{-1}} & \cdots & x^{e_r(p-1)^{-1}}\\
\vdots & \vdots & \ddots & \vdots \\
x^{e_1(p-i)^{-1}} & x^{e_2(p-i)^{-1}} & \cdots & x^{e_r(p-i)^{-1}}\\
\end{bmatrix} \begin{bmatrix}
\bar{s}_{e_1}(x)\\ \bar{s}_{e_2}(x)\\ \vdots \\ \bar{s}_{e_r}(x)\\
\end{bmatrix}\nonumber \\
=&\begin{bmatrix}
\sum_{u=1,u\neq e_2-e_1,\ldots,e_r-e_1}^{p-1}\bar{s}_{u +e_1}(x)\\
\sum_{u=1,u\neq e_2-e_1,\ldots,e_r-e_1}^{p-1}x^{u+e_1 }\bar{s}_{u+e_1}(x)\\
\vdots \\
\sum_{u=1,u\neq e_2-e_1,\ldots,e_r-e_1}^{p-1}x^{(u+e_1) (r-i-1)^{-1}}\bar{s}_{u+e_1}(x)\\
\sum_{u=1,u\neq e_2-e_1,\ldots,e_r-e_1}^{p-1}x^{(u+e_1)(p-1)^{-1} }\bar{s}_{u+e_1}(x)\\
\vdots \\
\sum_{u=1,u\neq e_2-e_1,\ldots,e_r-e_1}^{p-1}x^{(u+e_1) (p-i)^{-1}}\bar{s}_{u+e_1}(x)\\
\end{bmatrix},
\label{eq:thm-rec1}
\end{align}
\end{small}
for $i=1,2,\ldots,r-2$ and
\begin{small}
\begin{align}
&\begin{bmatrix}
1 & 1 & \cdots & 1\\
x^{e_1(p-1)^{-1}} & x^{e_2(p-1)^{-1}} & \cdots & x^{e_r(p-1)^{-1}}\\
\vdots & \vdots & \ddots & \vdots \\
x^{e_1(p-r+1)^{-1}} & x^{e_2(p-r+1)^{-1}} & \cdots & x^{e_r(p-r+1)^{-1}}\\
\end{bmatrix} \begin{bmatrix}
\bar{s}_{e_1}(x)\\ \bar{s}_{e_2}(x)\\ \vdots \\ \bar{s}_{e_r}(x)\\
\end{bmatrix}\nonumber \\
=&\begin{bmatrix}
\sum_{u=1,u\neq e_2-e_1,\ldots,e_r-e_1}^{p-1}\bar{s}_{u +e_1}(x)\\
\sum_{u=1,u\neq e_2-e_1,\ldots,e_r-e_1}^{p-1}x^{(u+e_1)(p-1)^{-1}}\bar{s}_{u+e_1}(x)\\
\vdots \\
\sum_{u=1,u\neq e_2-e_1,\ldots,e_r-e_1}^{p-1}x^{(u+e_1) (p-r+1)^{-1}}\bar{s}_{u+e_1}(x)\\
\end{bmatrix},
\label{eq:thm-rec2}
\end{align}
\end{small}
when $i=r-1$.
If the left $r\times r$ matrices in Eq. \eqref{eq:thm-rec1} and
Eq.~\eqref{eq:thm-rec2} are invertible in $\mathbb{F}_q[x]/(1+x+\cdots+x^{p-1})$, then we can compute the erased $r$ polynomials.
Therefore, the necessary and sufficient condition for recovering any $r$ erased lines of slope $i$
is proved.
\end{proof}

In Theorem \ref{thm:reco-tau1}, we presented a necessary and sufficient condition
for recovering any $r$ erased lines of slope $i$ for $0\leq i\leq r-1$. We have
$p$ distinct slopes in the $p\times (k+r)$ array; however the method in Theorem \ref{thm:reco-tau1}
is not applicable to the slope $i$ with $r\leq i\leq p-1$. The reason is as follows.
When $0\leq i\leq r-1$, we represent the erased $k+r$ symbols in each erased line by
a polynomial which is in $\mathcal{C}_{p}(1,1,q,d)$. According to the parity-check matrix
in Eq. \eqref{eq:matrixH2}, we can formulate $r$ linear equations of the $r$ erased
polynomials in $\mathcal{C}_{p}(1,1,q,d)$. If the corresponding $r\times r$ matrix
of the $r$ linear equations is invertible in $\mathbb{F}_q[x]/(1+x+\cdots+x^{p-1})$,
then we can recover the $r$ erased polynomials, i.e., the $r$ erased lines.
When $r\leq i\leq p-1$, the erased $r$ polynomials are in $\mathbb{F}_q[x]/(1+x^{p})$,
we are not sure whether each erased polynomial is in $\mathcal{C}_{p}(1,1,q,d)$ or not.
Therefore, there are many solutions in solving the $r$ linear equations of the
$r$ erased polynomials, even if the $r\times r$ matrix
of the $r$ linear equations is invertible in $\mathbb{F}_q[x]/(1+x+\cdots+x^{p-1})$.
Finding necessary and sufficient conditions for recovering any $r$ erased lines of slope
$i$ for $r\leq i\leq p-1$ is an open problem.

By Theorem \ref{thm:reco-tau1}, when $r=1,2,3$, we can check that the codes
$\textsf{GEBR}(p,1,k,r,q,1)$ can recover any $r$ erased
lines of slope $i$ for $0\leq i\leq r-1$, which is also shown by Theorem~40 in \cite{BR2019}.
When $r\geq 4$, we need to check that the matrix given in Theorem \ref{thm:reco-tau1}
is invertible over $\mathbb{F}_{q}[x]/(1+x+\ldots+x^{p-1})$.
Note that when the $r$ erased lines $e_1,e_2,\ldots,e_r$ are consecutive integers modulo $p$, i.e.,
$e_{i+1}=e_i+1\bmod p$ for $i=1,2,\ldots,r-1$, then the $r\times r$
matrix is a Vandermonde matrix
and is invertible over $\mathbb{F}_{q}[x]/(1+x+\ldots+x^{p-1})$. We can directly obtain the following
corollary from Theorem \ref{thm:reco-tau1}.
\begin{corollary}
The codes $\textsf{GEBR}(p,\tau,k,r,q,1)$ can recover any $r$ erased lines $e_1,e_2,\ldots,e_r$
with $e_{i+1}=e_i+1\bmod p$ of slope $i$ for $0\leq i\leq r-1$,
where $k+r\leq p$.
\label{thm:reco-tau1-2}
\end{corollary}

\begin{example}
Consider the code $\textsf{GEBR}(p=11,\tau=1,k=7,r=4,q,1)$. We have $k+r=11$ polynomials
\[
s_j(x)=s_{0,j}+s_{1,j}x+s_{2,j}x^2+\ldots+s_{10,j}x^{10},
\]
where $j=0,1,\ldots,10$ and
\[
s_{10,j}=s_{0,j}+s_{1,j}+s_{2,j}+\ldots+s_{9,j}.
\]
The parity-check matrix of the code is
\[
\begin{bmatrix}
1 & 1 &1 & 1 & 1 & 1 &1 & 1 & 1 & 1 & 1 \\
1 & x &x^2 &x^3 &x^4 &x^5 &x^6 &x^7 & x^{8} & x^{9} & x^{10}\\
1 & x^{2} &x^{4} &x^{6} &x^{8} &x^{10} &x^{12} &x^{14} & x^{16} & x^{18} & x^{20}\\
1 & x^{3} &x^{6} & x^{9} &x^{12} & x^{15} &x^{18} & x^{21} &x^{24} & x^{27} & x^{30}\\
\end{bmatrix}.
\]
According to row $\ell$ of the parity-check matrix, we have
\begin{equation}
s_{i,0}+s_{i-\ell,1}+s_{i-2\ell,2}+\cdots+s_{i-10\ell,10}=0,
\label{eq:gebr-exa-rec1}
\end{equation}
where $\ell=0,1,2,3$ and $i=0,1,\ldots,10$. Note that the indices are operated modulo $p=11$ in the example. Suppose that the following 44 symbols in four lines of slope 1
\begin{align*}
&s_{0,0},s_{10,1},s_{9,2},s_{8,3},s_{7,4},s_{6,5},s_{5,6},s_{4,7},s_{3,8},s_{2,9},s_{1,10},\\
&s_{1,0},s_{0,1},s_{10,2},s_{9,3},s_{8,4},s_{7,5},s_{6,6},s_{5,7},s_{4,8},s_{3,9},s_{2,10},\\
&s_{2,0},s_{1,1},s_{0,2},s_{10,3},s_{9,4},s_{8,5},s_{7,6},s_{6,7},s_{5,8},s_{4,9},s_{3,10},\\
&s_{3,0},s_{2,1},s_{1,2},s_{0,3},s_{10,4},s_{9,5},s_{8,6},s_{7,7},s_{6,8},s_{5,9},s_{4,10},
\end{align*}
are erased. For $i=0,1,\ldots,10$, we represent the following 11 symbols
\begin{align*}
&s_{i,0},s_{i-1,1},s_{i-2,2},s_{i-3,3},s_{i-4,4},s_{i-5,5},\\
&s_{i-6,6},s_{i-7,7},s_{i-8,8},s_{i-9,9},s_{i-10,10},
\end{align*}
in the line of slope 1 by the polynomial
\begin{equation}
\bar{s}_i(x)=s_{i,0}+s_{i-1,1}x+s_{i-2,2}x^2+\cdots+s_{i-10,10}x^{10}.
\label{eq:gebr-exm-rec2}
\end{equation}
By row $\ell=1$ (the second row) of the parity-check matrix, we have that the summation of all the coefficients of $\bar{s}_i(x)$ is zero. Therefore, $\bar{s}_i(x)\in \mathcal{C}_{p}(1,1,q,d)$.
We need to recover four polynomials $\bar{s}_0(x),\bar{s}_1(x),\bar{s}_2(x),\bar{s}_3(x)$ from the other 7 polynomials. By Eq. \eqref{eq:gebr-exa-rec1} with $\ell=0$, we have
\[
s_{i,0}+s_{i,1}+s_{i,2}+\cdots+s_{i,10}=0,
\]
where $i=0,1,\ldots,10$. Recall that the polynomial $\bar{s}_i(x)$ is given in Eq. \eqref{eq:gebr-exm-rec2}, we have
\[
\bar{s}_{0}(x)+x^{10}\bar{s}_{1}(x)+x^{9}\bar{s}_{2}(x)+\cdots+x\bar{s}_{10}(x)=0.
\]
When $\ell=2,3$, with the same argument, we have
\begin{align*}
&\bar{s}_{0}(x)+x\bar{s}_{1}(x)+x^{2}\bar{s}_{2}(x)+\cdots+x^{10}\bar{s}_{10}(x)=0,\\
&\bar{s}_{0}(x)+x^6\bar{s}_{1}(x)+x^{12}\bar{s}_{2}(x)+\cdots+x^{60}\bar{s}_{10}(x)=0.
\end{align*}
Since the summation of the 11 symbols in each column is zero, we have
\[
\bar{s}_{0}(x)+\bar{s}_{1}(x)+\bar{s}_{2}(x)+\cdots+\bar{s}_{10}(x)=0.
\]
Therefore, we obtain
\begin{align*}
\begin{bmatrix}
1 & x^{10} & x^9 & x^{8}\\
1 & 1 & 1 & 1\\
1 & x & x^2 & x^{3}\\
1 & x^6 & x^{12} & x^{18}\\
\end{bmatrix}\cdot \begin{bmatrix}
\bar{s}_{0}(x)\\\bar{s}_{1}(x)\\ \bar{s}_{2}(x) \\ \bar{s}_{3}(x)\\
\end{bmatrix}
=\begin{bmatrix}
\sum_{i=4}^{10}x^{11-i}\bar{s}_{i}(x)\\
\sum_{i=4}^{10}\bar{s}_{i}(x)\\
\sum_{i=4}^{10}x^{i}\bar{s}_{i}(x)\\
\sum_{i=4}^{10}x^{6i}\bar{s}_{i}(x)\\
\end{bmatrix}.
\end{align*}
Since
\begin{align*}
&\det\begin{bmatrix}
1 & x^{10} & x^9 & x^{8}\\
1 & 1 & 1 & 1\\
1 & x & x^2 & x^{3}\\
1 & x^6 & x^{12} & x^{18}\\
\end{bmatrix}\bmod (1+x^{11})\\
=&(1+x)(1+x^6)(1+x^{10})(x+x^6)(x+x^{10})(x^6+x^{10}),
\end{align*}
which is relatively prime to $1+x+\ldots+x^{10}$. Therefore, we can solve $\bar{s}_{0}(x),\bar{s}_{1}(x), \bar{s}_{2}(x), \bar{s}_{3}(x)$. Specifically, we can compute the four polynomials by Algorithm \ref{alg:lu}.
\end{example}

We have checked that the codes $\textsf{GEBR}(p,\tau,k,r=4,q,1)$ can recover
any $r=4$ erased lines $e_1,e_2,e_3,e_4$ with $0\leq e_1<e_2<e_3<e_4\leq p-1$ of slope
$i$ for $0\leq i\leq 3$, when $p=7,11,13,19$.

We can not employ the above method for the case of $\tau\geq 2$ in general, as the polynomial representing the $k+r$ erased symbols in the line of a slope is not in $\mathcal{C}_{p\tau}(1,\tau,q,d)$. We show in the next theorem that the code $\textsf{GEBR}(p,\tau,k,r,q,1)$ can recover up to $\tau$ erased lines of a slope, when $\tau\geq 2$. Note that the erased $\tau$ lines are not arbitrary.
\begin{theorem}
The code $\textsf{GEBR}(p,\tau,k,r,q,1)$ can recover any $\tau$ erased lines $e_1,e_2,\ldots,e_\tau$ of slope $i$ for $i=0,1,\ldots,r-1$, where $\tau \nmid (e_\alpha -e_\beta)$ for $\alpha\neq \beta\in \{1,2,\ldots,\tau\}$.
\label{thm:reco-taup-sloper}
\end{theorem}
\begin{proof}
%Similar to the proof of $\tau=1$, we can represent the erased $r$ lines as $r$ polynomials and we need to solve a $r\times r$ linear system over a quotient ring.
Suppose that $\tau$ lines $e_1,e_2,\ldots,e_\tau$ of slope $i$ are erased, where $\tau \nmid (e_\alpha -e_\beta)$ for $\alpha\neq \beta\in \{1,2,\ldots,\tau\}$.
For $\ell=e_1,e_2,\ldots,e_\tau$, the erased $k+r$ symbols in the line of slope $i$ are
$s_{\ell,0},s_{\ell-i,1},s_{\ell-2i,2},\ldots,s_{\ell-(k+r-1)i,k+r-1}$.
By Lemma \ref{eq:coeff}, we have
\[
s_{\ell,j}=\sum_{\mu=1}^{p-1}s_{\mu\tau +\ell,j},
\]
where $\ell=0,1,\ldots,p\tau-1$ and $j=0,1,\ldots,k+r-1$.
For any $\ell=e_1,e_2,\ldots,e_\tau$ and $j=0,1,\ldots,k+r-1$, the symbol
$s_{\ell-ji,j}$ is erased and the other $p-1$ symbols
$s_{\ell-ji+\tau,j},s_{\ell-ji+2\tau,j},\ldots,$ $s_{\ell-ji+(p-1)\tau,j}$
are not erased, as $\tau \nmid (e_\alpha -e_\beta)$ for $\alpha\neq \beta\in \{1,2,\ldots,\tau\}$.
Therefore, we can recover the erased symbol $s_{\ell-ji,j}$ by
\[
s_{\ell-ji,j}=s_{\ell-ji+\tau,j}+s_{\ell-ji+2\tau,j}+\ldots + s_{\ell-ji+(p-1)\tau,j}.
\]
\end{proof}

By Theorem \ref{thm:reco-taup-sloper}, $\textsf{GEBR}(p,\tau,k,r,q,1)$ can recover up to $\tau$ specified erased lines of slope $i$, for general parameter $r$.
In the following, we consider the code with $r=2$ and $\tau\geq 2$.

\begin{theorem}
If $\tau\geq 2$ and $p\tau>2(k+r-1)$, then the code $\textsf{GEBR}(p,\tau,k,r=2,q,1)$ can recover any two erased lines of slope $i$ for $i=0,1$.
\label{thm:reco-tau-slope01}
\end{theorem}
\begin{proof}
Recall that $s_{i,j}$ is the entry in row $i$ and column $j$ of the array in $\textsf{GEBR}(p,\tau,k,2,q,1)$, where $i=0,1,\ldots,p\tau-1$ and $j=0,1,\ldots,k+r-1$. Suppose that two lines of slope 0 are erased, i.e., rows $\alpha$ and $\beta$ of the array are erased, where $0\leq \alpha <\beta\leq p\tau-1$.
We need to recover $s_{\alpha,j}$ and $s_{\beta,j}$ for $j=0,1,\ldots,k+r-1$ from the other symbols.

As $r=2$, according to Eq. \eqref{eq:matrixH2}, we have that
\begin{equation}
\sum_{j=0}^{k+r-1}s_{i,j}=0 \text{ for } i=0,1,\ldots,p\tau-1,
\label{eq:slope0}
\end{equation}
and
\begin{equation}
\sum_{j=0}^{k+r-1}s_{i-j,j}=0 \text{ for } i=0,1,\ldots,p\tau-1.
\label{eq:slope1}
\end{equation}
If $\tau \nmid (\beta-\alpha)$, then we can recover $s_{\alpha,j}$ and $s_{\beta,j}$ by
\begin{equation}
\begin{array}{ll}
s_{\alpha,j}=&\sum_{\ell=1}^{p-1}s_{\ell \tau+\alpha,j},\\
s_{\beta,j}=&\sum_{\ell=1}^{p-1}s_{\ell \tau+\beta,j},
\end{array}
\label{eq:local-repair}
\end{equation}
according to Eq. \eqref{eq:coeff}.

Next, we assume that $\tau \mid (\beta-\alpha)$.
By Eq. \eqref{eq:slope1} with $i=\alpha$ and $i=\beta$, we have
\begin{align*}
&s_{\alpha,0}+s_{\alpha-1,1}+\ldots+s_{\alpha-(k+r-1),k+r-1}=0,\\
&s_{\beta,0}+s_{\beta-1,1}+\ldots+s_{\beta-(k+r-1),k+r-1}=0.
\end{align*}
If $0\leq \alpha-(k+r-1)$, or $0\geq \alpha-(k+r-1)$ and $\beta< p\tau+\alpha-(k+r-1)$, then we can recover $s_{\alpha,0}$ by
\[
s_{\alpha,0}=s_{\alpha-1,1}+s_{\alpha-2,2}+\ldots+s_{\alpha-(k+r-1),k+r-1}.
\]
Otherwise, we have $\beta\geq p\tau+\alpha-(k+r-1)$, then $\beta-(k+r-1)\geq p\tau+\alpha-2(k+r-1)>\alpha$ by assumption and we can recover $s_{\beta,0}$ by
\[
s_{\beta,0}=s_{\beta-1,1}+s_{\beta-2,2}+\ldots+s_{\beta-(k+r-1),k+r-1}.
\]
Once $s_{\alpha,0}$ or $s_{\beta,0}$ is known, we can recover  $s_{\beta,0}$ or  $s_{\alpha,0}$ by Eq. \eqref{eq:local-repair} with $j=0$. By repeating the above procedure for $i=\alpha+1,\alpha+2,\ldots,\alpha+k+r-1$ and $i=\beta+1,\beta+2,\ldots,\beta+k+r-1$, we can recover all $2n$ symbols $s_{\alpha,j}$ and $s_{\beta,j}$ for $j=0,1,\ldots,k+r-1$.

Next, we assume that two lines of slope 1 are erased, i.e., $s_{\alpha,0},s_{\alpha-1,1},\ldots,s_{\alpha-(k+r-1),k+r-1}$ and $s_{\beta,0},s_{\beta-1,1},\ldots,s_{\beta-(k+r-1),k+r-1}$ are erased. If $\tau \nmid (\beta-\alpha)$, we can recover $s_{\alpha-j,j}$ and $s_{\beta-j,j}$ by
\begin{equation}
\begin{array}{ll}
s_{\alpha-j,j}=&\sum_{\ell=1}^{p-1}s_{\ell \tau+\alpha-j,j},\\
s_{\beta-j,j}=&\sum_{\ell=1}^{p-1}s_{\ell \tau+\beta-j,j},
\end{array}
\label{eq:local-repair1}
\end{equation}
within the column. When $\tau \mid (\beta-\alpha)$, we can recover the symbol $s_{\alpha-(k+r-1),k+r-1}$ by
\[
s_{\alpha-(k+r-1),k+r-1}=\sum_{j=0}^{k+r-2}s_{\alpha-(k+r-1),j}
\]
if $\alpha-(k+r-1)\geq 0$ or $\alpha-(k+r-1)<0$ and $p\tau+\alpha-(k+r-1)>\beta$, or recover the symbol $s_{\beta,0}$ by
\[
s_{\beta,0}=s_{\beta,1}+s_{\beta,2}+\ldots+s_{\beta,k+r-1}
\]
if $\beta>p\tau+\alpha-(k+r-1)$.
Once $s_{\alpha-(k+r-1),k+r-1}$ or $s_{\beta,0}$ is known, we can recover $s_{\beta-(k+r-1),k+r-1}$ or $s_{\alpha,0}$ by Eq.~\eqref{eq:local-repair} with $j=k+r-1$ or $j=0$. Similarly, we can recover all $2n$ symbols in the erased two lines of slope 1.
\end{proof}

The next theorem shows that $\textsf{GEBR}(p,\tau,k,r=2,q,1)$ can recover more than two erased lines of slope $i$, if $\tau$ is large.
\begin{theorem}
If $\tau>k+r-1$, then the code $\textsf{GEBR}(p,\tau,k,r=2,q,1)$ can recover any three erased lines of slope $i$ for $i=0,1$.
\label{thm:reco-tau-slope013}
\end{theorem}
\begin{proof}
Suppose that rows $\alpha$, $\beta$ and $\gamma$ of the array are erased, where $0\leq \alpha <\beta<\gamma\leq p\tau-1$.
We want to recover $s_{\alpha,j}$, $s_{\beta,j}$ and $s_{\gamma,j}$ for $j=0,1,\ldots,k+r-1$ from the other symbols.

According to Eq. \eqref{eq:matrixH2}, we have that
\begin{equation}
\begin{array}{ll}
&\sum_{j=0}^{k+r-1}s_{i,j}=0 \text{ for } i=0,1,\ldots,p\tau-1,\\
&\sum_{j=0}^{k+r-1}s_{i-j,j}=0 \text{ for } i=0,1,\ldots,p\tau-1.
%&\sum_{j=0}^{k+r-1}s_{i-2j,j}=0 \text{ for } i=0,1,\ldots,p\tau-1.
\end{array}
\label{eq:slope012}
\end{equation}
If $\tau \nmid (\beta-\alpha)$ and $\tau \nmid (\gamma-\beta)$, then we can recover $s_{\alpha,j}$, $s_{\beta,j}$ and $s_{\gamma,j}$ by
\begin{equation}
\begin{array}{ll}
s_{\alpha,j}=&\sum_{\ell=1}^{p-1}s_{\ell \tau+\alpha,j},\\
s_{\beta,j}=&\sum_{\ell=1}^{p-1}s_{\ell \tau+\beta,j},\\
s_{\gamma,j}=&\sum_{\ell=1}^{p-1}s_{\ell \tau+\gamma,j},
\end{array}
\label{eq:local-repair012}
\end{equation}
according to Eq. \eqref{eq:coeff}. If $\tau \mid (\beta-\alpha)$ and $\tau \nmid (\gamma-\beta)$, then we can first recover $s_{\gamma,j}$ by Eq. \eqref{eq:local-repair012} and then recover $s_{\alpha,j}$ and $s_{\beta,j}$ by Theorem \ref{thm:reco-tau-slope01}. Similarly, we can recover the erased symbols if $\tau \nmid (\beta-\alpha)$ and $\tau \mid (\gamma-\beta)$.

In the following, we assume that $\tau \mid (\beta-\alpha)$ and $\tau \mid (\gamma-\beta)$.
By Eq. \eqref{eq:slope012}, we have
\begin{align*}
&s_{\alpha+j,0}+s_{\alpha+j-1,1}+\ldots+s_{\alpha-(k+r-1)+j,k+r-1}=0,\\
&s_{\beta+j,0}+s_{\beta+j-1,1}+\ldots+s_{\beta-(k+r-1)+j,k+r-1}=0,\\
&s_{\gamma+j,0}+s_{\gamma+j-1,1}+\ldots+s_{\gamma-(k+r-1)+j,k+r-1}=0,
\end{align*}
where $j=0,1,\ldots,k+r-1$.
As $\tau>k+r-1$, we have
\begin{align*}
&\alpha +j<\beta, 0<\alpha-(k+r-1)+j \text{ or } \\
&0>\alpha-(k+r-1)+j \text{ and } p\tau+\alpha-(k+r-1)+j>\gamma,\\
&\alpha<\beta-(k+r-1)+j \text{ and } \gamma>\beta+j,\\
&\beta<\gamma-(k+r-1)+j \text{ and } p\tau+\alpha>\gamma+j,
\end{align*}
for $j=0,1,\ldots,k+r-1$
and we can recover $s_{\alpha,j},s_{\beta,j},s_{\gamma,j}$ by
\begin{align*}
s_{\alpha,j}=&\sum_{i=0,i\neq j}^{k+r-1}s_{\alpha+j-i,i},\\
s_{\beta,j}=&\sum_{i=0,i\neq j}^{k+r-1}s_{\beta+j-i,i},\\
s_{\gamma,j}=&\sum_{i=0,i\neq j}^{k+r-1}s_{\gamma+j-i,i}.
\end{align*}
Similarly, we can recover all $2n$ symbols in any erased three lines of slope 1.
\end{proof}
The recovery of erased lines of slope $i$ in
$\textsf{GEBR}(p,\tau,k,r,q,1)$ for general parameters $\tau$ and $r$
is an open problem and is a subject of future work.

\section{Comparison with LRC and Product Codes}
\label{sec:com}
LRCs \cite{2014A} and product codes \cite{gopalan2017maximally} are two
families of existing codes that can locally repair
any single-symbol failure. The differences between our GEBR codes and the existing two codes
are as follows.

Given $k\alpha$ information symbols, an $(m(k+r),\alpha k,k+r)$ LRC \cite{2014A} creates $r\alpha$ global
parity symbols by encoding all the information symbols, divides
all $(k+r)\alpha$ symbols (including $k\alpha$ information symbols and $r\alpha$
global parity symbols) into $k+r$ groups which are placed into $k+r$ columns and obtains $m-\alpha$ local
parity symbols for each group. Compared with LRC, our GEBR codes have two advantages.
First, each symbol in our GEBR codes can be repaired
by either some symbols in the same column or the symbols along each of $r$ lines of slope,
while a symbol in LRC can only be locally repaired within the group. Second, our codes
have much lower decoding complexity. When there are $r$ column failures, we can formulate
$r$ linear equations for the erased $r$ columns with encoding matrix being Vandermonde matrix
for GEBR codes and solve the erased $r$ columns by the proposed fast LU decoding algorithm.
While the $r$ linear equations corresponding to the $r$ erased columns (groups) for the LRC
are not Vandermonde linear equations, there is no fast decoding algorithm designed for the LRC
when $r$ columns have failed. Moreover, although we can obtain some well-designed LRC that can
recover some $r$ erased lines of a slope, the underlying field size should be large enough.

Note that LRCs with availability \cite{2013Repair,2015Bounds} are special LRCs such that each
symbol can be repaired with multiple disjoint repair groups. However, the construction of
LRCs with availability \cite{2013Repair,2015Bounds} to achieve the known bound on symbol distance require a sufficiently large
field, and therefore incur much more decoding complexity than the proposed codes when some lines
are erased.

A product code with parameters $k,r,\alpha,m$ organizes the $k\alpha$ information symbols
into an $\alpha \times k$ information array, first creates $r$ local parity symbols for each
row and then obtains $m-\alpha$ local parity symbols for each of the $k+r$ columns.
Therefore, any symbol can be recovered by accessing some symbols in the same row or in the
same column, but not in a line of a non-zero slope. Second, product code can only recover
 at most any $m-\alpha$ row failures but our GEBR codes can recover  at most any $\max\{m-\alpha,r\}$ row
failures. Finally, the minimum symbol distance of product code is at most
$(m-\alpha+1)(r+1)$, while the minimum symbol distance of GEBR codes is strictly larger
than $(m-\alpha+1)(r+1)$ for some parameters.

\begin{table*}[tbh]
\caption{Comparison with $(p^2,(p-1)(p-r),p)$ LRCs and product codes with $m=p$, $\alpha=p-1$.}
\begin{center}
\begin{tabular}{|c|c|c|c|c|c|} \hline
% Column 0 &Column 1 & Column 2 & Column 3 & Column 4& Column 5& Column 6& Column 7& Column 8   \\ \hline \hline
Codes  & single-symbol & $r$-column failures & $r$ consecutive-row  & $r$ consecutive-row &minimum symbol  \\
  & failures &  decoding& failures  & failures decoding & distance  \\ \hline
GEBR  &each of $r+1$ lines &fast decoding & yes&fast decoding &$\geq 2r+2$ \\ \hline
LRC  &one line (the same group) &no fast decoding & maybe over large field&no &$(p-1)r+2$ \\ \hline
Product   &each of two lines &fast decoding & no&N/A &$2r+2$ \\ \hline
\end{tabular}
\end{center}
\label{table:com}
%\vspace{-0.5cm}
\end{table*}

Table \ref{table:com} shows the comparison of our GEBR codes, LRCs and product codes, when
$m=p$ and $\alpha=p-1$. It is easy to check that the three codes have the same storage overhead.
When any single-symbol fails, GEBR codes have $r+1$ disjoint repair groups, while LRCs and product codes
have one and two repair groups, respectively. In decoding $r$-column failures, both GEBR codes and
product codes have fast decoding algorithm, there is no fast decoding algorithm for LRCs.
We can decode any $r$ consecutive-row failures for GEBR codes by the fast LU decoding algorithm.
Although it is possible to design LRCs over a large finite field to recover any $r$
consecutive-row failures, there is no fast decoding algorithm for general parameters.
LRCs have the largest minimum symbol distance among the three codes. Compared with LRCs,
GEBR codes can be viewed as codes with larger recoverability for single-symbol failures,
multi-column failures and multi-row failures, at a cost of minimum symbol distance reduction.
When compared with product codes, GEBR codes not only have larger recoverability, but also
possible have larger minimum symbol distance for some parameters.

Consider the code $\textsf{GEBR}(p=11,\tau=1,k=7,r=4,q,1)$ in Example \ref{alg:lu},
we have $k(p-1)=70$ data symbols and $p^2-k(p-1)=51$ local parity symbols. Each symbol has
$r+1=5$ disjoint repair groups. We can recover
any $r=4$ erased lines $e_1,e_2,e_3,e_4$ with $0\leq e_1<e_2<e_3<e_4\leq 10$ of slope
$i$ for $0\leq i\leq 3$. We can also recover any four column failures by the fast LU
decoding algorithm. While for the product code with the same parameters, we can only
recover any symbol by two disjoint repair groups, and we can not recover any four erased
lines.

\section{Conclusion}
\label{sec:con}
In this paper, we propose a coding method of array codes that has local repair property.
We present the constructions of GEBR codes and GEIP codes based on the proposed coding method that can support much more parameters, compared with EBR codes and EIP codes, respectively.
We propose an efficient LU decoding method for GEBR codes and GEIP codes based on the LU factorization of Vandermonde matrix.
When $\tau$ is large, we show that GEBR codes have both larger minimum symbol distance and larger recovery ability of erased lines for some parameters, compared with EBR codes.
The $(n,k)$ recoverable condition of GEBR codes for general $g(x)$ is one of our future work.
How to propose a coding framework to unify GEBR codes, LRCs, and product codes is another future work.
It is also interesting to explore some good properties by replacing each column
of the proposed codes with regenerating codes.

\appendices
\section{Proof of Lemma \ref{lm:div}}
\label{pr:lm-div}
We first show that $r_0=\sum_{u=1}^{\frac{p-1}{2}}\sum_{\ell=1}^{\tau}f_{(2u-1)\tau b+\ell b}$.
According to Eq. \eqref{eq:div}, we have
\begin{equation}
r_{\tau b+\ell b }=r_{\tau b+(\ell-1)b}+f_{\tau b+\ell b},
\label{eq:lm-div1}
\end{equation}
where $\ell=0,1,\ldots,p\tau-1$. Summing both sides of Eq. \eqref{eq:lm-div1} from $\ell=0$ to $\ell=(i-1)\tau$, we have
\begin{equation}
r_{i b\tau}=r_{\tau b-b}+\sum_{\ell=0}^{(i-1)\tau}f_{\tau b+\ell b},
\label{eq:lm-div2}
\end{equation}
where $i=1,2,\ldots,p-1$. Summing both sides of Eq. \eqref{eq:lm-div2} from $i=1$ to $i=p-1$, we have
\begin{equation}
\sum_{i=1}^{p-1}r_{i b\tau}=\sum_{i=1}^{p-1}r_{i\tau}=(p-1)r_{\tau b-b}+\sum_{i=1}^{p-1}\sum_{\ell=0}^{(i-1)\tau}f_{\tau b+\ell b},
\label{eq:lm-div3}
\end{equation}
where the first equation above comes from that $\gcd (b,p)=1$.
By Eq. \eqref{lm:coeff} in Lemma \ref{lm:coeff}, we have $\sum_{i=1}^{p-1}r_{i\tau}=r_0$.
Since $p$ is an odd prime number, we have $(p-1)r_{\tau b-b}=0$. We can compute $\sum_{i=1}^{p-1}\sum_{\ell=0}^{(i-1)\tau}f_{\tau b+\ell b}$ as
\begin{align*}
&\sum_{i=1}^{p-1}\sum_{\ell=0}^{(i-1)\tau}f_{\tau b+\ell b}=(p-1)f_{\tau b}+(p-2)\sum_{\ell=1}^{\tau}f_{\tau b+\ell b}+\\
&(p-3)\sum_{\ell=1}^{\tau}f_{2\tau b+\ell b}+\cdots+2\sum_{\ell=1}^{\tau}f_{(p-3)\tau b+\ell b}+\sum_{\ell=1}^{\tau}f_{(p-2)\tau b+\ell b}\\
=&\sum_{u=1}^{\frac{p-1}{2}}\sum_{\ell=1}^{\tau}f_{(2u-1)\tau b+\ell b}.
\end{align*}
Therefore, we obtain that $r_0=\sum_{u=1}^{\frac{p-1}{2}}\sum_{\ell=1}^{\tau}f_{(2u-1)\tau b+\ell b}$. Similarly, we can show that Eq. \eqref{eq:div1} holds for $j=0,1,\ldots,a-1$.
Once $r_0$ is known, we can compute other $\frac{p\tau}{a}-1$ coefficients recursively by Eq. \eqref{eq:div2} with $j=0$ and $\ell=1,2,\ldots,\frac{p\tau}{a}-1$. Similarly, we can compute $\frac{p\tau}{a}-1$ coefficients recursively by Eq. \eqref{eq:div2} with $\ell=1,2,\ldots,\frac{p\tau}{a}-1$ for $j=0,1,\ldots,a-1$, after solving $r_j$.

Next, we need to show that the solved $r(x)$ is in $\mathcal{C}_{p\tau}(g(x),\tau,q,d)$, i.e., $g(x)(1+x^{\tau})$ divides $r(x)$. First, $(1+x^{\tau})$ divides $r(x)$, as we can show that $\sum_{\ell=0}^{p-1}r_{\ell\tau+\mu}=0$ for $\mu=0,1,\ldots,\tau-1$. Second, since $g(x)$ divides $f(x)$, if $\gcd (1+x^b,g(x))=1$, then $g(x)$ divides $r(x)$. As $\gcd (g(x),1+x)=1$ and $\gcd (p,b)=1$, we have that $\gcd (1+x^b,g(x))=1$ by Lemma 19 in \cite{BR2019}. Therefore, $g(x)(1+x^{\tau})$ divides $r(x)$ and the lemma is proved.

\section{Proof of Lemma \ref{lm:div2}}
\label{pr:lm-div2}
Since $\gcd (b,m)=\gcd (up^s,p^{\nu+1})=p^s$ and $\gcd (u,p)=1$,
we have that $\gcd (u,p^{\nu+1})=1$.
In the following, we show that
\begin{eqnarray}
&&\{0,up^s,2up^s,\cdots,u(p^{\nu+1}-2p^s)\}\bmod p^{\nu+1}\nonumber\\
&=&\{0,p^s,2p^s,\cdots,p^{\nu+1}-2p^s\}.
\label{eq:lm-div4}
\end{eqnarray}
First, we prove that if $i\neq j\in\{0,p^s,2p^s,\cdots,p^{\nu+1}-2p^s\}$,
then $u\cdot i\neq u\cdot j \bmod p^{\nu+1}$. Suppose that
$u\cdot i= u\cdot j \bmod p^{\nu+1}$, then there exists an
integer $\ell$ such that
\[
u\cdot i= u\cdot j +\ell p^{\nu+1},
\]
and we can further obtain that
\[
u\cdot (i-j)= \ell p^{\nu+1}.
\]
Since $\gcd (u,p^{\nu+1})=1$, we have $p^{\nu+1}\mid (i-j)$,
which contradicts to that $i\neq j\in\{0,p^s,2p^s,\cdots,p^{\nu+1}-2p^s\}$.
Similarly, we can show that
\[
u\cdot i\neq  p^{\nu+1}-p^s.
\]
Therefore, Eq. \eqref{eq:lm-div4} holds.
According to Eq. \eqref{eq:div}, we have
\begin{equation}
f_{2iup^s+up^s}=r_{2iup^s+up^s}+r_{2iup^s},
\label{eq:lm-div21}
\end{equation}
where $i=0,1,\ldots,p^{\nu-s+1}-1$.
Summing both sides of Eq.~\eqref{eq:lm-div21} from $i=0$ to $i=\frac{p^{\nu-s+1}-3}{2}$, we have
\begin{eqnarray}
\sum_{i=0}^{\frac{p^{\nu-s+1}-3}{2}}f_{2iup^s+up^s}
&=&\sum_{i=0}^{\frac{p^{\nu-s+1}-3}{2}}(r_{2iup^s+up^s}+
r_{2iup^s})\nonumber \\
&=&\sum_{i=0}^{p^{\nu-s+1}-2}r_{iup^s}\nonumber \\
&=&\sum_{i=0}^{p^{\nu-s+1}-2}r_{ip^s}\label{eq:lm-div-p1} \\
&=&r_{(p^{\nu-s+1}-1)p^s}=r_{p^{\nu+1}-p^s},
\label{eq:lm-div22}
\end{eqnarray}
where Eq. \eqref{eq:lm-div-p1} comes from Eq. \eqref{eq:lm-div4},
Eq. \eqref{eq:lm-div22} comes from that
\begin{align*}
&\{0,p^s,2p^s,\cdots,p^{\nu+1}-2p^s\}=\\
&\left\{\!\!
\begin{array}{l}
\{0,p^{\nu},\cdots,(p-1)p^{\nu}\}\cup
\{p^s,p^{\nu}+p^s,\cdots,(p-1)p^{\nu}+p^s\}\\
\cup \cdots \cup \{2p^{\nu}-p^s,3p^{\nu}-p^s,\cdots,(p-1)p^{\nu}-p^s\}, \text{ if } \nu>s,\\
\{0,p^{\nu},2p^{\nu},\cdots,(p-2)p^{\nu}\},  \text{ if } \nu=s.
\end{array}
\right.
\end{align*}
Similarly, we can show that Eq. \eqref{eq:div21} holds for
$j=0,1,\ldots,m-1$.
Once $r_{p^{\nu+1}-p^s+j}$ for $j=0,1,\ldots,p^s-1$ are known, we can compute the other coefficients recursively.

Recall that $\sum_{i=0}^{\frac{p^{\nu-s+1}-3}{2}}f_{2iup^s+up^s+j}
=r_{p^{\nu+1}-p^s+j}$
for $j=0,1,\ldots,m-1$ by Eq. \eqref{eq:lm-div22}, we have
\[
r_{i}=f_{up^s+p^s+i}+f_{3up^s+p^s+i}+\cdots+f_{(p^{\nu-s+1}-2)up^{s}+p^s+i}
\]
for $i=0,1,\ldots,m$. Recall that the indices are taken modulo $m=p^{\nu+1}$.
We have
\[
r(x)=(x^{p^{\nu+1}-up^s-p^s}+x^{p^{\nu+1}-3up^s-p^s}+\cdots+x^{2up^s-p^s})f(x).
\]
Since $f(x)\in\mathcal{C}_{p\tau}(g(x),\tau,q,d)$, we have that
$r(x)\in\mathcal{C}_{p\tau}(g(x),\tau,q,d)$ and the lemma is proved.

\bibliographystyle{IEEEtran}
%\bibliography{references}

\vspace{-7mm}
\begin{IEEEbiographynophoto}{Hanxu Hou} received the B.Eng. degree in
Information Security from Xidian University, Xian, China, in 2010, and Ph.D. degrees
in the Dept. of Information Engineering from the Chinese University of Hong Kong
in 2015 and in the School of Electronic and Computer Engineering from
Peking University in 2016. He is now an Associate Professor with Dongguan University of
Technology. He was a recipient of the 2020 Chinese Information Theory Young Rising Star
Award by China Information Theory Society. He was recognized as an Exemplary Reviewer
2020 in IEEE Transactions on Communications. His research interests include
erasure coding and coding for distributed storage systems.
\end{IEEEbiographynophoto}

\begin{IEEEbiographynophoto}{Yunghsiang S. Han}
(S'90-M'93-SM'08-F'11) was born in Taipei, Taiwan, 1962. He received B.Sc. and M.Sc. degrees in electrical engineering from the National Tsing Hua University, Hsinchu, Taiwan, in 1984 and 1986, respectively, and a Ph.D. degree from the School of Computer and Information Science, Syracuse University, Syracuse, NY, in 1993. He was from 1986 to 1988 a lecturer at Ming-Hsin Engineering College, Hsinchu, Taiwan. He was a teaching assistant from 1989 to 1992, and a research associate in the School of Computer and Information Science, Syracuse University from 1992 to 1993. He was, from 1993 to 1997, an Associate Professor in the Department of Electronic Engineering at Hua Fan College of Humanities and Technology, Taipei Hsien, Taiwan. He was with the Department of Computer Science and Information Engineering at National Chi Nan University, Nantou, Taiwan from 1997 to 2004. He was promoted to Professor in 1998. He was a visiting scholar in the Department of Electrical Engineering at University of Hawaii at Manoa, HI from June to October 2001, the SUPRIA visiting research scholar in the Department of Electrical Engineering and Computer Science and CASE center at Syracuse University, NY from September 2002 to January 2004 and July 2012 to June 2013, and the visiting scholar in the Department of Electrical and Computer Engineering at University of Texas at Austin, TX from August 2008 to June 2009. He was with the Graduate Institute of Communication Engineering at National Taipei University, Taipei, Taiwan from August 2004 to July 2010. From August 2010 to January 2017, he was with the Department of Electrical Engineering at National Taiwan University of Science and Technology as Chair Professor. From February 2017 to February 2021, he was with School of Electrical Engineering \& Intelligentization at Dongguan University of Technology, China. Now he is with the Shenzhen Institute for Advanced Study, University of Electronic Science and Technology of China. He is also a Chair Professor at National Taipei University from February 2015. His research interests are in error-control coding, wireless networks, and security.

Dr. Han was a winner of the 1994 Syracuse University Doctoral Prize and a Fellow of IEEE. One of his papers won the prestigious 2013 ACM CCS Test-of-Time Award in cybersecurity.
\end{IEEEbiographynophoto}

\begin{IEEEbiographynophoto}{Patrick P. C. Lee} received the B.Eng. degree (first class honors) in Information Engineering from the Chinese University of Hong Kong in 2001, the M.Phil.
degree in Computer Science and Engineering from the Chinese University
of Hong Kong in 2003, and the Ph.D. degree in Computer Science from
Columbia University in 2008. He is now a Professor of the
Department of Computer Science and Engineering at the Chinese University
of Hong Kong. His research interests are in various applied/systems topics
including storage systems, distributed systems and networks, operating systems, dependability, and security.
\end{IEEEbiographynophoto}

\begin{IEEEbiographynophoto}{You Wu}
received the M.Phil. degree from Guangdong University of Technology in 2021. She is now
in Beijing Didi Infinity Technology and Development Co., Ltd. Her research interests
include the coding for distributed storage systems.
\end{IEEEbiographynophoto}

\begin{IEEEbiographynophoto}{Guojun Han}
received the M.E. degree from South China University of Technology, Guangzhou, China, and
the Ph.D. degree from Sun Yatsen University, Guangzhou, China. From March 2011 to August 2013,
he was a Research Fellow at the School of Electrical and Electronic Engineering,
Nanyang Technological University, Singapore. From October 2013 to April 2014, he was a
Research Associate at the Department of Electrical and Electronic Engineering, Hong Kong
University of Science and Technology. He is now a Full Professor and Executive Dean at
the School of Information Engineering, Guangdong University of Technology, Guangzhou, China.

He has been a Senior Member of IEEE since 2014. His research interests are in the areas of wireless communications, signal processing, coding and information theory. He has more than 15 years¡¯ experience on research and development of advanced channel coding and signal processing algorithms and techniques for various data storage and communication systems.
\end{IEEEbiographynophoto}

\begin{IEEEbiographynophoto}{Mario Blaum}
(Life Fellow, IEEE) was born in Buenos Aires, Argentina.
He received the Licenciado degree from the University of Buenos Aires in
1977, the M.Sc. degree from the Technion¡ªIsrael Institute of Technology
in 1981, and the Ph.D. degree from the California Institute of Technology
(Caltech) in 1984, all in mathematics.

In 1985, he was a Research Fellow at the Department of Electrical
Engineering, Caltech. In 1985, he joined the IBM Research Division, Almaden
Research Center. In 2003, his division was transferred to Hitachi Global
Storage Technologies, where he was a Research Staff Member until 2009,
in which he rejoined the IBM Almaden Research Center. Since 2001, he has
been an Academic Advisor at the Universidad Complutense of Madrid, Spain.
He retired in 2021. His research interest includes all aspects of coding for
storage technology.
\end{IEEEbiographynophoto}

\end{document}